\documentclass[acmsmall,screen,nonacm]{acmart}
\setcopyright{none}
\acmJournal{TELO}
\acmVolume{0}
\acmNumber{0}
\acmArticle{0}
\acmYear{2026}
\acmMonth{1}
\acmDOI{}

\usepackage{tikz}
\usetikzlibrary{positioning}
\usepackage{subcaption}   % replaces subfig, which clashes with acmart
\usepackage{multirow}
\usepackage{bm}
\usepackage{threeparttable}
\makeatletter
\AtBeginDocument{%
  \@ifundefined{th@acmplain}{\theoremstyle{plain}}{\theoremstyle{acmplain}}%
  \@ifundefined{theorem}{\newtheorem{theorem}{Theorem}[section]}{}%
  \@ifundefined{lemma}{\newtheorem{lemma}[theorem]{Lemma}}{}%
  \@ifundefined{corollary}{}{}%
  \@ifundefined{proposition}{\newtheorem{proposition}[theorem]{Proposition}}{}%
  \@ifundefined{conjecture}{}{}%
  \@ifundefined{th@acmdefinition}{\theoremstyle{definition}}{\theoremstyle{acmdefinition}}%
  \@ifundefined{definition}{\newtheorem{definition}[theorem]{Definition}}{}%
  \@ifundefined{example}{}{}%
  \@ifundefined{remark}{\newtheorem{remark}[theorem]{Remark}}{}%
}
\makeatother

\begin{document}

\title{The Optimal Sampler for Diverse TSP Solutions: Characterization, Limits, and Learning}

%% Authors are suppressed automatically by the "anonymous" option, but
%% keep the blocks filled in so the camera-ready needs no extra work.
\author{Ssu-Yuan Lo}
\email{lawrence090402@gmail.com}
\affiliation{%
  \institution{National Central University}
  \department{Department of Computer Science and Information Engineering}
  \country{Taiwan}}

\author{Ching-Kai Wang}
\email{a910424a@gmail.com}
\affiliation{%
  \institution{National Central University}
  \department{Department of Computer Science and Information Engineering}
  \country{Taiwan}}

\author{Kuan-Lun Chen}
\email{glchen000000@gmail.com}
\affiliation{%
  \institution{National Central University}
  \department{Department of Computer Science and Information Engineering}
  \country{Taiwan}}

\author{Hao-Tsung Yang}
\email{htyang@g.ncu.edu.tw}
\affiliation{%
  \institution{National Central University}
  \department{Department of Computer Science and Information Engineering}
  \country{Taiwan}}

\renewcommand{\shortauthors}{Lo et al.}

\begin{abstract}
We address the Diverse Traveling Salesman Problem (D-TSP) that seeks a set of $k$ distinct TSP tours. Each selected tour must not exceed a prescribed multiple of the optimal tour length, while the maximum pairwise Jaccard similarity within the selected set is minimized. This formulation is important for applications requiring both solution quality and fault tolerance, such as logistics planning, robotics, and strategic patrolling. We propose a novel hybrid framework that decomposes D-TSP into a sampling stage and a selection stage. The selection stage is a max-min dispersion problem, for which any approximation better than a factor of two is NP-hard; we therefore focus on the sampling stage. Our contributions are both theoretical and empirical. Theoretically, we give two end-to-end guarantees for the optimal sampling distribution, from two perspectives: including an optimal set in the sampled pool, or having at least $k$ mutually distant tours in the pool. The second reduces the number of samples from exponential to linear in $k$, at the cost of weakening the benchmark from the best $k$ tours to the best $\beta k$ for some $\beta > 1$. For practical use, we prove that a pipeline that samples spanning trees and then converts to tours does not hurt the quality of the pool in principle. Empirically, both objectives are learnable, and our experiments not only validate the analysis but also outperform the baselines on both in-distribution and out-of-distribution instances, by a single trained model with no fine-tuning.
\end{abstract}

%% ACM requires CCS concepts. Generate your own at
%% https://dl.acm.org/ccs/ccs_flat.cfm and paste the XML it gives you.
\begin{CCSXML}
<ccs2012>
<concept>
<concept_id>10003752.10003809.10003716.10011136.10011797.10011799</concept_id>
<concept_desc>Theory of computation~Approximation algorithms analysis</concept_desc>
<concept_significance>500</concept_significance>
</concept>
<concept>
<concept_id>10002950.10003624.10003633.10010917</concept_id>
<concept_desc>Mathematics of computing~Combinatorial optimization</concept_desc>
<concept_significance>500</concept_significance>
</concept>
<concept>
<concept_id>10010147.10010257.10010293.10011809.10011813</concept_id>
<concept_desc>Computing methodologies~Reinforcement learning</concept_desc>
<concept_significance>300</concept_significance>
</concept>
</ccs2012>
\end{CCSXML}

\ccsdesc[500]{Theory of computation~Approximation algorithms analysis}
\ccsdesc[500]{Mathematics of computing~Combinatorial optimization}
\ccsdesc[300]{Computing methodologies~Reinforcement learning}

\keywords{Traveling salesman problem, solution diversity, dispersion,
max--min packing, graph neural networks, reinforcement learning}

\received{}
\received[revised]{}
\received[accepted]{}

\maketitle

% Text of your paper here
%IMPROVE: 1. argue that most work on max-sum, wheras we're working on max-min

\section{Introduction}\label{sec:Intro}
The demand for a set of diverse solutions rather than a single one is the premise of quality-diversity and niching methods in evolutionary computation, which deliberately maintain a population spread across a behavioural or genotypic space instead of converging on a single optimum~\citep{pugh2016quality,
cully2017quality}, and it has been taken up for combinatorial problems in particular. Work in this direction falls into two families. In the first, diversity is of interest in its own right, as in diverse minimum $s$--$t$ cuts~\citep{de2023finding, gao2022obtaining}, stable matchings~\citep{deberg2025findingdiversesolutionscombinatorial}, and vertex cover~\citep{zhou2018efficient}. In the second, diversity is the means by which solution quality is secured, as in the travelling thief problem~\citep{nikfarjam2024use} and batch selection in surrogate optimization~\citep{nezami2025enhancing}. On the application side, a diversity objective also serves as a powerful heuristic across a range of domains: it
helps locate near-optimal solutions~\citep{qian2024quality, kim2021learning}; it reduces predictability while preserving the quality of a mixed strategy in adversarial settings~\citep{basilico2022recent, yang2026patrol,
sinha2018stackelberg}; it improves the robustness and stability of learned policies~\citep{yang2025towards, ganesh2025systemizing, pang2019improving,
Jiang_Wu_Cao_Zhang_2022}; and it supplies the interpretability that opaque black-box models otherwise lack~\citep{mothilal2020explaining, dandl2020multi,
chen2025pareto, nguyen2023feasible}.

In this work we address the Diverse Traveling Salesman Problem (D-TSP). That is, given an edge-weighted graph and a constant $c$, the objective of D-TSP is to identify a set $S$ of $k$ distinct TSP tours. This set must satisfy two competing criteria: first, the length of every tour $T \in S$ must be at most $c|T^*|$, where $|T^*|$ is the length of an optimal TSP tour; and second, the maximum pairwise Jaccard similarity among tours in the set $S$ must be minimized. This formulation seeks high-quality solutions while maximizing their structural dissimilarity.

Within this growing theme, Diverse TSP also has many compelling potential applications. These include optimization tasks in logistics, where redundant and distinct path options are necessary for fault tolerance, and in robotics, such as efficient patrol planning and task allocation where various routes must be covered~\citep{yang2019patrol,liu2019task}. Despite this, the development of a guaranteed approximation algorithm for this problem is still open. The most closely related work are from the authors in~\citep{de2025price}, which focuses on 2 diverse tours only, and the authors in~\citep{gao2022obtaining}, which addresses the bi-criteria optimization of generating $k$ diverse approximate spanning trees. However, converting this approach as TSP tours suffers two critical limitations: First, the quality of the diversity achieved for spanning trees may fall apart when the trees are converted into Hamiltonian tours (see Figure \ref{fig:counter_examples_combined} as an example). Secondly, the diverse solutions in~\citep{gao2022obtaining} relied on budget constraint programming, which faces the scalability issue for large graph. 

\textbf{Contributions.} This paper studies D-TSP through a sample-then-select
framework. We make the following contributions.

\begin{itemize}
  \item \textbf{Our first guarantee is a $2$-approximation.} With high probability, any sampler that is uniform over the feasible region and places no mass outside it returns,
after $M$ draws, the pairwise Jaccard distance is at least half of the optimal. The required $M$ is $N_c \ln(k/\delta)$, where $N_c$ is the number of feasible tours, which is exponential in the number of cities in general. 
  \item Any samplers without the sampling distribution described above pays a further factor according to our analysis. This factor is unavoidable: we show a family of Euclidean instances on which a non-uniform sampler needs proportionally more samples than the uniform one essentially.

  \item For practical use, we provide a pipeline that samples spanning trees and converts to tours. We show that this relaxation costs nothing in principle but attributable to the objective alone. This residue is measurable from a batch.

\item By connecting D-TSP to the maximum independent set problem, \textbf{our second guarantee requires a sample size linear in $k$ in expectation.} Write $p_\mu(\rho)$ for the probability that two independently drawn tours lie within
distance $\rho$ of each other. Fix $\beta > 1$ and a threshold
$0 < \rho \le D^{*}_{\beta k}(\Omega_c)$, and suppose the sampler satisfies $p_\mu(\rho) \le 1/(\beta k)$. Then $M = O(k)$ draws already contain, in expectation, a set of $k$ tours whose pairwise Jaccard distances are all at
least $\rho$. Certifying the same conclusion with high probability rather than in expectation costs nothing further in samples, but it does bound $\beta$ away from one.

  \item Our experiments
  confirm the analysis above, and show that our method beats three
  state-of-the-art methods on every instance, both on in-distribution random
  40-city instances and on out-of-distribution ones (TSPLIB instances with
  51--101 cities), when the budget matches the one it was trained for.
\end{itemize}

\begin{figure}[htbp]
    \centering
    
    % --- 第一張子圖 (Tree 1) ---
    \subcaptionbox{Counter example tree-1}{%
        \includegraphics[width=0.31\linewidth]{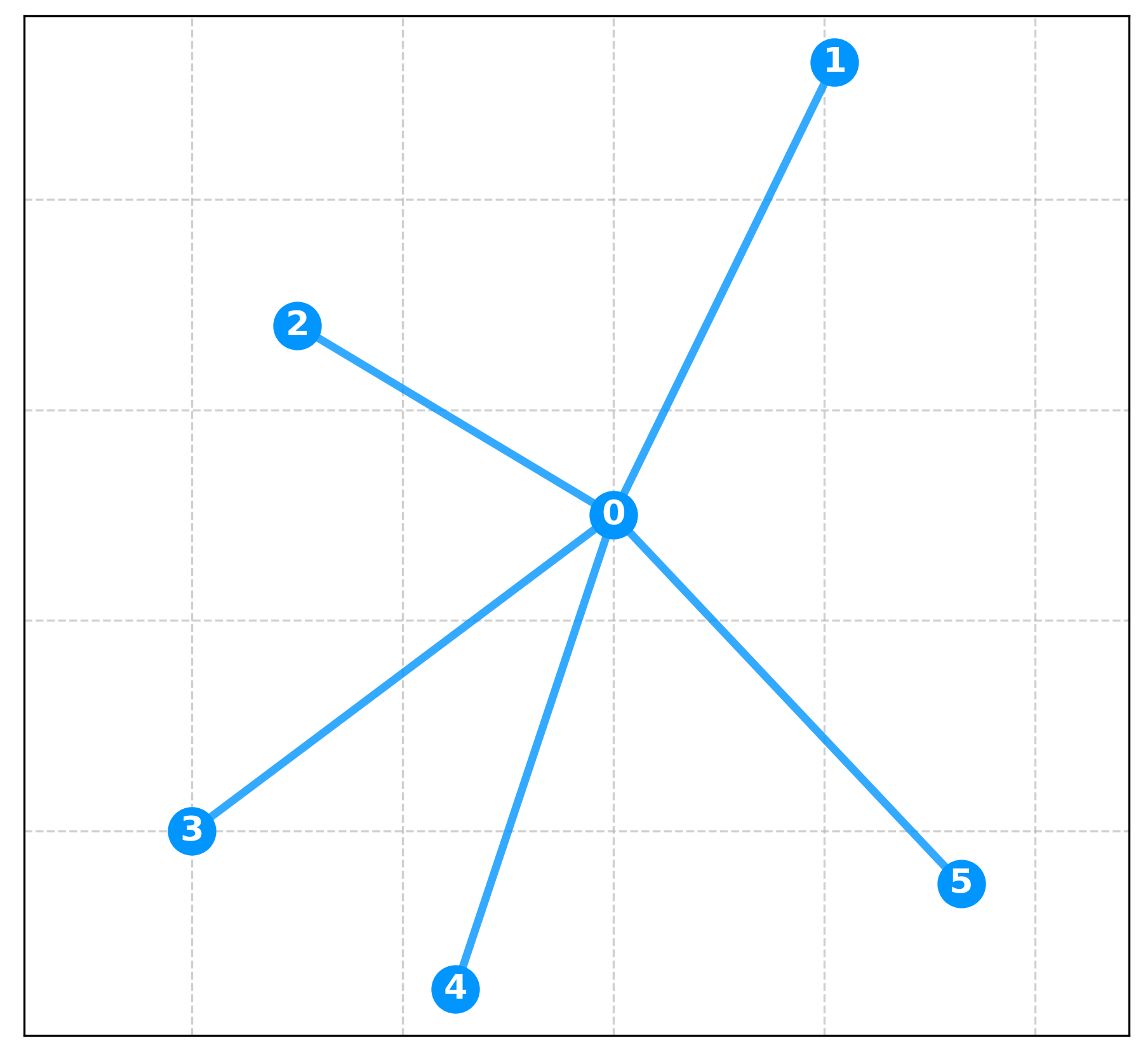}%
        \label{fig:counter_example_tree-1}%
    }
    \hfill % 撐開間距
    % --- 第二張子圖 (Tree 2) ---
    \subcaptionbox{Counter example tree-2}{%
        \includegraphics[width=0.31\linewidth]{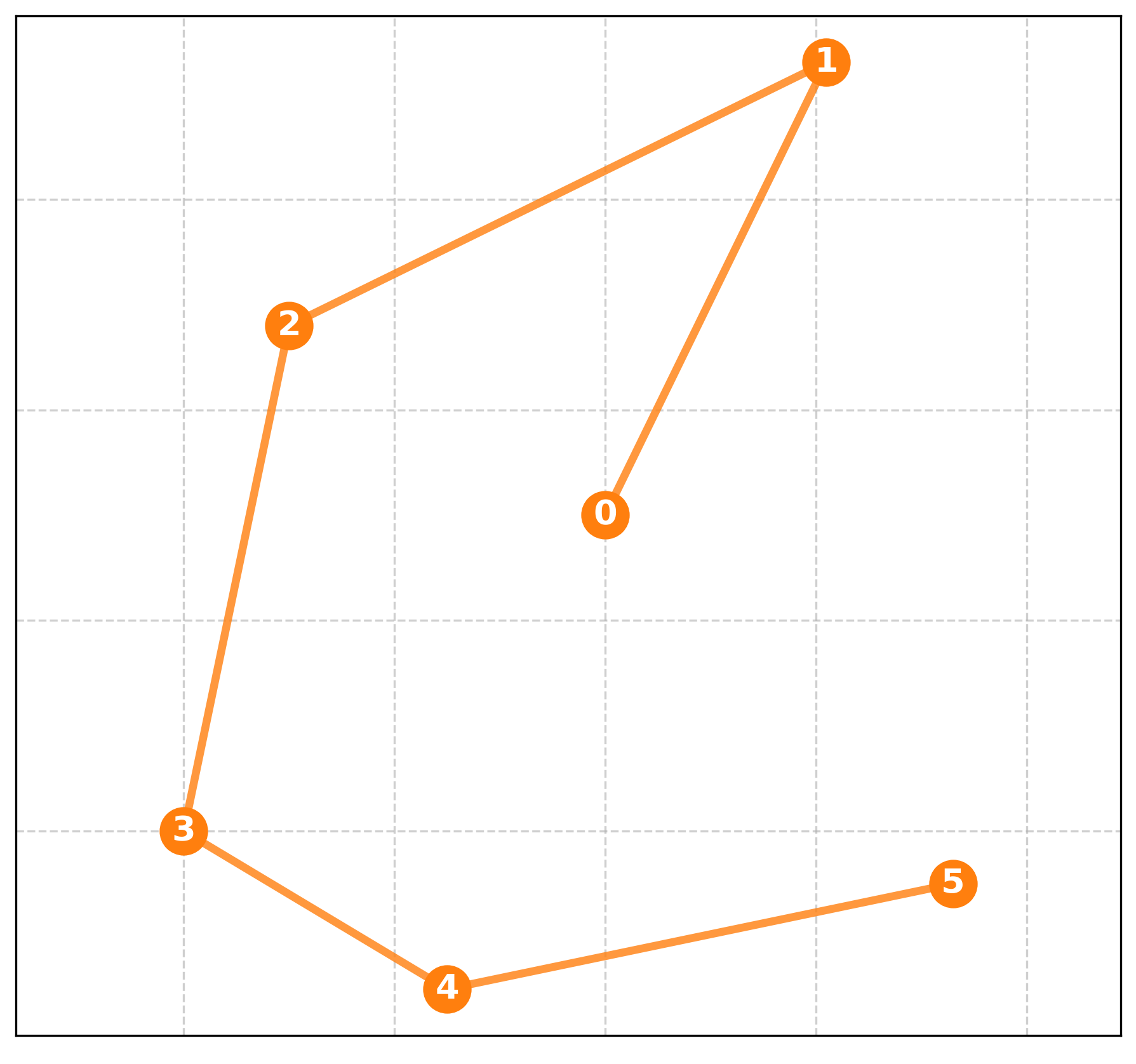}%
        \label{fig:counter_example_tree-2}%
    }
    \hfill % 撐開間距
    % --- 第三張子圖 (Tour) ---
    \subcaptionbox{Tour}{%
        \includegraphics[width=0.31\linewidth]{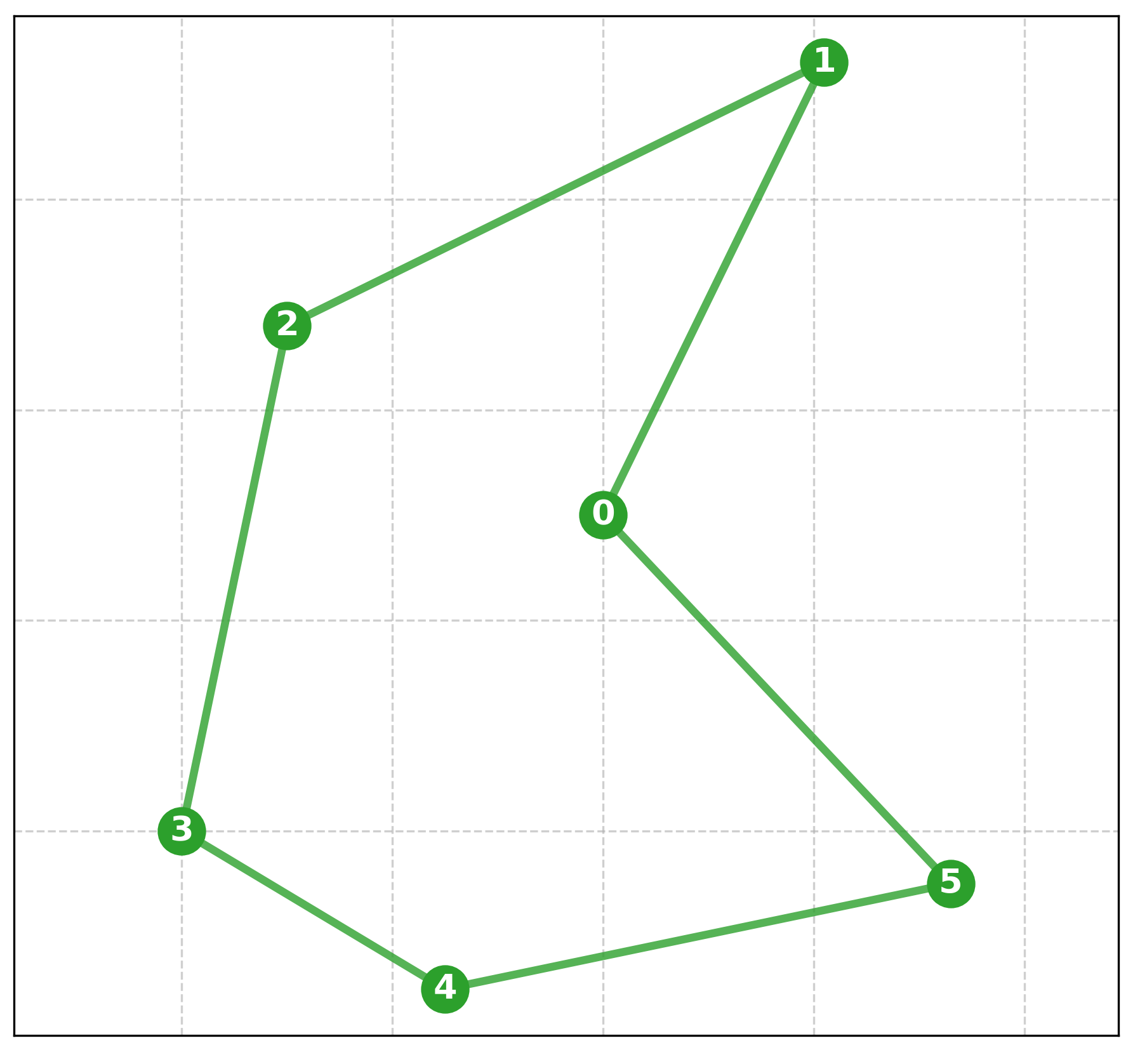}%
        \label{fig:counter_example_tour}%
    }
    
    \caption{A counterexample shows that one cannot preserve the dissimilarity when converting spanning trees into tours. (a) and (b) depict two distinct spanning trees sharing exactly one edge (resulting in a near-zero Jaccard index). Despite this minimal overlap, (c) demonstrates that distinct spanning trees can generate the exact same TSP tour via the double-tree heuristic.}
    \label{fig:counter_examples_combined}
\end{figure}

\section{Related Work}\label{sec:Related Work}
Existing research addresses this challenge in heuristic ways, which can generally be divided into two main categories: traditional heuristics for Multi-Solution TSP (MSTSP) and neural approaches for Diverse TSP.

\textbf{Traditional Heuristics for MSTSP.} The Multi-Solution TSP has primarily been tackled by traditional heuristics~\citep{angus2006niching, do2022niching, han2018multimodal, huang2019niching, huang2018seeking, ronald1995finding}. Although advanced genetic algorithms have shown exceptional performance in finding a single optimal tour for the standard TSP~\citep{nagata2013powerful}, adapting them to efficiently explore multiple diverse solutions remains computationally challenging. Most of these methods are based on genetic or evolution-based algorithms, which naturally incorporate diverse "genes" or solution components during the evolution phase. The main difference among these lines of research is their focus on how to effectively explore and maintain the search space while simultaneously preserving the diversity and quality of the generated solutions. A significant drawback in this family of work is the high computation time required, as the problems scale up due to their NP-hard and multi-solution nature. Notably, the Niching Memetic Algorithm (NMA)~\citep{huang2019niching} has been identified as a state-of-the-art traditional heuristic for this problem~\citep{liu2023evolutionary}. NMA utilizes niche techniques that initially divide the generated TSP tours into distinct groups, maintaining tours only if they are similar to others within their assigned niche, thus protecting different regions of the solution space. Following NMA, approaches based on Evolutionary Diversity Optimization (EDO) have also been proposed~\citep{do2022niching}. EDO frameworks typically integrate diversity-specific metrics directly into the fitness function or use specialized evolution operators to push the population toward distinct areas of the solution space. Unlike NMA, EDO methods initialize from a known optimal solution and evolve the population solely to maximize diversity, which can be impractical given the NP-hard property of finding the optimal TSP tour.

\textbf{Neural Approaches for Diverse TSP.} In contrast to traditional heuristics, neural approaches for combinatorial optimization (CO) learn the solutions in a data-driven way. While these algorithms require a training stage, they offer the benefit of high-speed inference thanks to the parallel structure afforded by GPUs. The key challenge is developing policies that can effectively transfer to other graphs with different sizes and topologies.

Most initial neural approaches focused on CO with single-criteria optimization. Early foundational evaluations explored continuous neural architectures like Hopfield networks for the standard TSP~\citep{potvin1993state}. Following the rise of deep learning, modern approaches shifted to sequence-to-sequence learning, such as the foundational work of the Pointer Network~\citep{vinyals2015pointer, yin2020enhancing}, Graph Pointer Networks~\citep{ma2019combinatorial, yang2022graph}, and the Attention Model~\citep{vaswani2017attention, kool2019attention}. More recent work has shifted focus to multi-solution generation. By utilizing multiple parallel decoders, the model can explore different regions of the solution space~\citep{xin2021multi}. A key example is the work from the authors in~\citep{li2025diversity}, which proposes a deep reinforcement learning solver specifically for Diversity Optimization for TSP. This model is equipped with several mechanisms such as a relativization filter, temperature softmax~\citep{yi2019sampling}, and active search~\citep{hottung2022efficient}. These additional mechanisms, while effective, make the underlying model architecture complicated and tend to decrease overall interpretability. In addition, the active search method, while increasing the quality of the solutions, slowing down the inference stage significantly for large number of cities. 

% =====================================================================
% Diverse TSP -- Section 3 (Problem Definition), revised
%
% Notation conventions used throughout the paper:
%   - d(.,.) is the Jaccard distance on edge sets and is the ONLY
%     distance symbol. The Hamming formulation is expressed through the
%     overlap a(.,.) instead of a second metric symbol.
%   - The objective MAXIMIZES the minimum pairwise distance, so the
%     2-approximation of Section 4 reads correctly.
%   - D_k^*(S) is written with the ground set as an argument, so that
%     Sections 4 and 6 can compare D_k^*(P_M) with D_k^*(Omega_c).
%   - C(.) is tour cost, so generic sets are named A, B, F.
%   - \mathbf{1} is the indicator (\mathbb{1} is not defined by amssymb
%     and will not compile).
%   - Table 1 lists symbols in the order: problem input, problem output,
%     optimum and near-optimum, tools used in the analysis.
% =====================================================================

\section{Problem Definition}
\label{sec:problem}

We formulate the problem of generating diverse TSP solutions within the
framework of the Multi-Solution Traveling Salesman Problem
\citep{li2025diversity}. The goal is to find a set of tours that are at once
near-optimal and structurally different from one another. Table~\ref{tab:notation}
collects the symbols used throughout the paper.

\begin{table}[t]
\centering
\caption{Notation used throughout the paper.}
\label{tab:notation}
\begin{tabular}{ll}
\hline
Symbol & Meaning \\
\hline
\multicolumn{2}{l}{\emph{Problem input}} \\
$n$ & number of cities \\
$\mathcal{T}_n$ & set of all tours, $|\mathcal{T}_n| = (n-1)!/2$ \\
$\pi$, $E_\pi$ & a tour and its edge set, $|E_\pi| = n$ \\
$C(\pi)$ & cost of a tour \\
$c \ge 1$ & cost tolerance \\
$k$ & number of tours requested \\
\multicolumn{2}{l}{\emph{Problem output}} \\
$\Pi$ & the returned set of $k$ tours \\
$a(\pi,\pi')$ & number of shared edges, an integer in $\{0,\dots,n\}$ \\
$d(\pi,\pi')$ & Jaccard distance, equal to $2(n-a)/(2n-a)$ \\
$D(\Pi)$ & smallest pairwise distance inside $\Pi$ \\
\multicolumn{2}{l}{\emph{Optimum and near-optimum}} \\
$C^{*}$ & optimal tour cost \\
$\widehat{C}$ & computable upper bound on $C^{*}$, used in its place \\
$\Omega_c$, $N_c$ & feasible region and its cardinality \\
$D^{*}_k(S)$ & largest $D(\Pi)$ over all $k$-subsets $\Pi$ of $S$ \\
\multicolumn{2}{l}{\emph{Tools used in the analysis}} \\
$\mu$ & sampling distribution on $\Omega_c$ \\
$M$, $\mathcal{P}_M$ & number of retained samples, and the pool they form \\
$\alpha$ & entropy coefficient (Section~\ref{sec:method}) \\
$\rho$ & separation threshold \\
$G_\rho(S)$ & conflict graph on $S$ at threshold $\rho$ \\
$\mathrm{ind}(\cdot)$ & size of a largest independent set \\
$p_\mu(\rho)$ & collision rate, $\Pr[d(X,Y) < \rho]$ \\
\hline
\end{tabular}
\end{table}

\subsection{Tours and cost}

Let $G = (V, E, w)$ be a complete weighted graph, where
$V = \{v_1,\dots,v_n\}$, and $w : E \to \mathbb{R}^{+}$ gives the Euclidean
distance between two vertices. A solution to the TSP is a tour $\pi$, that is,
a Hamiltonian cycle, which we identify with its edge set
$E_\pi \subseteq E$. We write $\mathcal{T}_n$ for the set of all tours and
record the elementary but consequential fact that
\begin{equation}
  |E_\pi| = n \qquad \text{for every tour } \pi ,
  \label{eq:equal-cardinality}
\end{equation}
so that the solution space is a family of \emph{equal-cardinality} edge sets.
Section~\ref{sec:method} uses this property to sharpen our guarantee and to
connect our objective to the Hamming-distance formulations common in the
diverse-solutions literature.

The cost of a tour is $C(\pi) = \sum_{e \in E_\pi} w(e)$, and $C^{*}$ denotes
the optimal cost, attained by some tour $\pi^{*}$.

\subsection{Feasible region and distance}

\begin{definition}[Feasible region]
\label{def:feasible}
Given a tolerance $c \ge 1$, the feasible region is
\begin{equation}
  \Omega_c \;=\; \{\, \pi \in \mathcal{T}_n \;:\; C(\pi) \le c\,C^{*} \,\},
  \qquad N_c = |\Omega_c| .
  \label{eq:feasible}
\end{equation}
\end{definition}

All tours in $\Omega_c$ are treated as equally acceptable: the cost constraint
is a \emph{hard filter}, not a secondary objective. This is what makes the
sampling question in Section~\ref{sec:method} well posed, since it says what a
sampler should aim at inside $\Omega_c$, namely nothing in particular.

\begin{remark}[Using $\widehat{C}$ in place of $C^{*}$]
\label{rem:chat}
Deciding whether a given tour lies in $\Omega_c$ already requires $C^{*}$, so
$\Omega_c$ is not computable. We therefore replace $C^{*}$ by an upper bound
$\widehat{C} \ge C^{*}$ returned by a strong heuristic solver, which enlarges
the feasible region to a superset of $\Omega_c$. Every statement in this paper
holds verbatim for the region defined by $\widehat{C}$, because none of them
uses optimality of $C^{*}$; only the meaning of the tolerance changes, from $c$
to the effective value $c\,\widehat{C}/C^{*}$. We write $C^{*}$ below for
readability and return to the substitution where it matters.
\end{remark}

To measure how similar two tours are we use the Jaccard index of their edge
sets and the distance it induces,
\begin{equation}
  J(\pi,\pi') \;=\; \frac{|E_{\pi} \cap E_{\pi'}|}{|E_{\pi} \cup E_{\pi'}|},
  \qquad
  d(\pi,\pi') \;=\; 1 - J(\pi,\pi') .
  \label{eq:jaccard}
\end{equation}
Write $a = a(\pi,\pi') = |E_{\pi} \cap E_{\pi'}|$ for the overlap. By
Equation~\ref{eq:equal-cardinality} we have $|E_{\pi} \cup E_{\pi'}| = 2n - a$,
so both quantities are functions of the overlap alone:
\begin{equation}
  J(\pi,\pi') = \frac{a}{2n-a},
  \qquad
  d(\pi,\pi') = \frac{2(n-a)}{2n-a} ,
  \label{eq:J-in-a}
\end{equation}
with $J$ strictly increasing and $d$ strictly decreasing in $a$. We use $d$ as
the only distance in this paper and express everything else through $a$.

The selection guarantee of Section~\ref{sec:method} needs $d$ to be a metric.
This is known but not obvious \citep{levandowsky1971distance, kosub2019note},
and we give the short proof because closely related measures fail it.\footnote{%
The Dice--S\o{}rensen dissimilarity $1 - 2a/(|E_\pi| + |E_{\pi'}|)$
\citep{dice1945measures} and the cosine dissimilarity are equally natural as
overlap measures and appear in the diversity literature, but neither satisfies
the triangle inequality, so Theorem~\ref{thm:greedy} would not apply to them.}

\begin{proposition}[$d$ is a metric]
\label{prop:metric}
On any family of non-empty finite sets, $d$ is a metric. In particular $d$ is a
metric on $\mathcal{T}_n$.
\end{proposition}

\begin{proof}
Non-negativity and symmetry are immediate, and $d(A,B) = 0$ forces
$|A \cap B| = |A \cup B|$, hence $A = B$. For the triangle inequality, let
$\sigma$ be a uniformly random permutation of the ground set and let
$h_\sigma(A) = \arg\min_{e \in A} \sigma(e)$. The smallest-rank element of
$A \cup B$ is uniform on $A \cup B$, and $h_\sigma(A) = h_\sigma(B)$ holds
exactly when that element lies in $A \cap B$, so
$\Pr_\sigma[h_\sigma(A) \neq h_\sigma(B)] = d(A,B)$. For any fixed $\sigma$,
equality is transitive, which gives the pointwise inequality
\[
  \mathbf{1}[h_\sigma(A) \neq h_\sigma(F)] \;\le\;
  \mathbf{1}[h_\sigma(A) \neq h_\sigma(B)] +
  \mathbf{1}[h_\sigma(B) \neq h_\sigma(F)] .
\]
Taking expectations over $\sigma$ yields $d(A,F) \le d(A,B) + d(B,F)$. Every
tour has $|E_\pi| = n \ge 3$, so $h_\sigma$ is well defined on tours.
\end{proof}

\begin{remark}[The distance is coarse]
\label{rem:resolution}
Two distinct tours cannot share $n-1$ edges, since the last edge is then forced
and the tours coincide; hence $a \le n-2$ and, by Equation~\ref{eq:J-in-a},
\[
  d(\pi,\pi') \;\ge\; \frac{4}{n+2}
  \qquad \text{for } \pi \neq \pi' .
\]
The bound is attained exactly by pairs that differ in two edges, that is, by a
single $2$-opt move. So $d$ takes only $n$ distinct values, and $4/(n+2)$ is
the finest resolution available on $\mathcal{T}_n$. Section~\ref{sec:method}
works below this resolution, where covering a tour means hitting it exactly.
\end{remark}

\subsection{The Diverse TSP}

\begin{definition}[D-TSP]
\label{def:dtsp}
For a set $\Pi$ of at least two tours write
$D(\Pi) = \min_{\pi \neq \pi' \in \Pi} d(\pi,\pi')$, and for a finite set of
tours $S$ with $|S| \ge k$ write
\begin{equation}
  D^{*}_k(S) \;=\; \max_{\Pi \subseteq S,\, |\Pi| = k} D(\Pi) .
  \label{eq:objective}
\end{equation}
Given $G$, a tolerance $c \ge 1$ and an integer $k \ge 2$, the Diverse
Traveling Salesman Problem asks for a set of $k$ tours in $\Omega_c$ attaining
$D^{*}_k(\Omega_c)$.
\end{definition}

Carrying the ground set as an argument is deliberate. Our method optimizes over
a sampled pool rather than over $\Omega_c$, and Sections~\ref{sec:method}
and~\ref{sec:conflict} both turn on the difference between $D^{*}_k(\mathcal{P}_M)$
and $D^{*}_k(\Omega_c)$.

For a selected set $\Pi$ we also write $J(\Pi) = 1 - D(\Pi)$ for its largest
pairwise Jaccard index. Since $d = 1 - J$, Definition~\ref{def:dtsp} is
equivalent to minimizing $J(\Pi)$, the form used in prior work
\citep{gao2022obtaining}, and we report $J(\Pi)$ in all experiments for
comparability. We state the problem in the distance form because the guarantee
of Section~\ref{sec:method} is multiplicative in $d$ and only additive in $J$;
stating it on the similarity scale would misrepresent its strength.

\begin{remark}[Connection to Hamming diversity]
\label{rem:hamming}
The normalized symmetric difference $|E_\pi \triangle E_{\pi'}|/(2n) = (n-a)/n$
is, like $d$, strictly decreasing in the overlap $a$. Any objective that
compares pairs of tours by distance therefore gives the same answer under
either quantity, so Definition~\ref{def:dtsp} is \emph{equivalent} to the
max--min Hamming formulation studied for diverse spanning trees, matchings and
matroid bases \citep{hanaka2021finding, hanaka2023framework, danna2009select}.
Section~\ref{sec:method} uses this equivalence once, to tighten what the
guarantee says on the Jaccard scale at no algorithmic cost.
\end{remark}

% =====================================================================
% Diverse TSP -- Section 4 (Sampling and Selection), revised
%
% Notation (see Table 1 in Section 3):
%   - d is the only distance symbol; overlaps are written through a.
%   - Delta(mu) is the only measure of non-uniformity. Lambda and the
%     covering number N_eps have been removed.
%   - There is no cover radius: by Remark on resolution, a ball of
%     radius below 4/(n+2) is a single point, so "covering" an optimal
%     tour means hitting it exactly, and q is a hitting probability.
%   - eta is the cost gap of the lower-bound instance and is used
%     nowhere else. The ring radius of that instance is rho, since R is
%     the reward.
%   - The construction and the two longer proofs are in the Appendix.
% =====================================================================

\section{Sampling and Selection}
\label{sec:method}

D-TSP (Definition~\ref{def:dtsp}) is challenging along two independent axes.
Certifying that a single tour lies in $\Omega_c$ already requires $C^{*}$, so
feasibility checking alone contains the TSP. Even when the feasible region is
given explicitly as a finite set, selecting the $k$ mutually most distant
elements is the max--min dispersion problem, which is NP-hard
\citep{ravi1994heuristic}. The feasible region is moreover of exponential size,
admits no compact description, and the objective couples all $\binom{k}{2}$
pairs, so direct optimization over $\Omega_c$ is impractical.

We therefore decompose the problem into two stages. The first stage draws tours
from a distribution and keeps those that pass the cost filter, forming a pool
$\mathcal{P}_M \subseteq \Omega_c$ of $M$ retained tours; throughout, $M$ counts
retained tours rather than raw draws. The second stage extracts $k$ tours from
the pool greedily: starting from $S_1 = \{\pi_1\}$ with $\pi_1$ arbitrary, for
$i = 2,\dots,k$,
\begin{equation}
  \pi_i \;=\; \arg\max_{\pi \in \mathcal{P}_M \setminus S_{i-1}}\;
              \min_{\pi_j \in S_{i-1}} d(\pi, \pi_j),
  \qquad S_i = S_{i-1} \cup \{\pi_i\},
  \label{eq:greedy}
\end{equation}
and returns $\Pi = S_k$. The second stage is governed by a classical guarantee,
which applies here because $d$ is a metric (Proposition~\ref{prop:metric}).

\begin{theorem}[\citealt{ravi1994heuristic}]
\label{thm:greedy}
Let $(X, d')$ be a finite metric space with $|X| \ge k$. The set returned by
Equation~\ref{eq:greedy} run on $X$ under $d'$ satisfies
$D(\Pi) \ge D^{*}_k(X)/2$. Moreover, no polynomial-time algorithm achieves any
factor better than $2$ for max--min dispersion in general metric spaces unless
P $=$ NP.
\end{theorem}

The hardness statement is for general metric spaces. Tours under $d$ form a very
particular family, so it does not follow that $2$ is optimal here, and whether
this family admits a better constant is open. What matters for us is the split
it induces. Assuming the pool holds at least $k$ distinct tours,
\begin{equation}
  \frac{D(\Pi)}{D^{*}_k(\Omega_c)}
  \;=\;
  \underbrace{\frac{D(\Pi)}{D^{*}_k(\mathcal{P}_M)}}_{\ge\, 1/2\ \text{by Theorem~\ref{thm:greedy}}}
  \;\times\;
  \underbrace{\frac{D^{*}_k(\mathcal{P}_M)}{D^{*}_k(\Omega_c)}}_{\text{no guarantee}} .
  \label{eq:decomposition}
\end{equation}
The first factor is fixed by a known algorithm with a known analysis. The second
is at most $1$, carries no guarantee at all, and depends only on which tours the
sampler put in the pool. The rest of this section is therefore about discussing and finding the optimal
sampler which can close this gap.

\subsection{Hitting an optimal set}
\label{subsec:hit}

Since the cost constraint of Equation~\ref{eq:feasible} is a hard filter, tour
length carries no information inside $\Omega_c$, and the sampler is free to
place its mass anywhere in the feasible region. We now ask where it should place
it. Fix an optimal set $\Pi^{*} = \{o_1,\dots,o_k\} \subseteq \Omega_c$
attaining $D^{*}_k(\Omega_c)$, and let $\mu$ be the sampler conditioned on
$\Omega_c$. We first consider what would be the chance to hit all the tours in an optimal set.

\begin{definition}[Hitting probability]
\label{def:hit}
$q(\mu;\Pi^{*}) \;=\; \min_{1 \le i \le k} \mu(o_i)$.
\end{definition}

\begin{lemma}[Sample size]
\label{lem:sample}
Draw $M$ tours i.i.d.\ from $\mu$. If $M \ge \ln(k/\delta)/q$ then with
probability at least $1-\delta$ the pool contains $\Pi^{*}$, so
$D^{*}_k(\mathcal{P}_M) = D^{*}_k(\Omega_c)$ and
$D(\Pi) \ge D^{*}_k(\Omega_c)/2$. Conversely, containing the lightest $o_i$ with
probability at least $1-\delta$ requires
$M \ge \ln(1/\delta) / \ln\bigl(1/(1-q)\bigr)$, which is
$(1-o(1))\ln(1/\delta)/q$ as $q \to 0$.
\end{lemma}

\begin{proof}
The probability that no draw equals $o_i$ is $(1-q)^M \le e^{-qM} \le \delta/k$,
and a union bound over the $k$ tours shows that all of them are drawn with
probability at least $1-\delta$. On that event $\mathcal{P}_M$ contains
$\Pi^{*}$, and since $\mathcal{P}_M \subseteq \Omega_c$ the two optima agree;
Theorem~\ref{thm:greedy} then gives the bound on $D(\Pi)$. For the converse,
$(1-q)^M \le \delta$ is necessary, and taking logarithms gives the stated bound.
\end{proof}

The cost of the design is therefore $\Theta(1/q)$, and nothing else about $\mu$
enters. Since $\Pi^{*}$ is unknown, we ask which sampler is best when it is
chosen adversarially, and measure a sampler by
\begin{equation}
  \Delta(\mu) \;=\; \bigl( N_c \min_{\pi \in \Omega_c} \mu(\pi) \bigr)^{-1} .
  \label{eq:delta}
\end{equation}

\begin{proposition}[The uniform sampler is optimal]
\label{prop:minimax}
$\sup_{\Pi^{*}} q(\mu;\Pi^{*})^{-1} = \Delta(\mu)\, N_c$, where the supremum is
over all $k$-subsets of $\Omega_c$. Moreover $\Delta(\mu) \ge 1$, with equality
if and only if $\mu$ is uniform on $\Omega_c$.
\end{proposition}

\begin{proof}
A set $\Pi^{*}$ containing a minimizer of $\mu$ attains the supremum, and no set
can do better, so $\sup_{\Pi^{*}} q^{-1} = (\min_\pi \mu(\pi))^{-1}$. For the
second claim, $\min_\pi \mu(\pi) \le 1/N_c$ because the values average to
$1/N_c$, with equality exactly when they are all equal.
\end{proof}

Combining the two statements gives the worst-case requirement in closed form:
\begin{equation}
  \sup_{\Pi^{*}} M \;=\; \Delta(\mu)\, N_c \ln(k/\delta) .
  \label{eq:sample-size}
\end{equation}

Thus, the uniform sampler pays $N_c \ln(k/\delta)$ and every other sampler pays
$\Delta$ times as much. Now, one may ask how to obtain a sampler whose
distribution is close to uniform on $\Omega_c$. The following examples show that reach it by rejection is
hopeless. 

In an instance of $n$ cities drawn uniformly from the unit square, a uniformly
random tour has expected cost $\Theta(n)$ while the optimum is only
$\Theta(\sqrt{n})$, so the ratio between the two grows like $0.73\sqrt{n}$ and
is already about $4.6$ at $n = 40$. For any fixed tolerance the acceptance
probability therefore decays as $e^{-\Theta(n)}$, so drawing tours at random and
discarding the infeasible ones is exponentially wasteful; Appendix~\ref{app:rejection}
gives the argument. The uniform target has to be learned rather than filtered
for.

The next proposition shows that learning it requires nothing beyond a correct
reward signal and an entropy bonus. A policy $p$ over tours is trained by
maximizing
\begin{equation}
  \mathbb{E}_{\pi \sim p}[R(\pi)] + \alpha H(p),
  \label{eq:obj}
\end{equation}
where $R(\pi)$ is the reward assigned to a tour,
$H(p) = -\sum_{\pi} p(\pi) \log p(\pi)$ is the entropy of the policy, and
$\alpha > 0$ is the entropy coefficient that weights the two terms. For $R$ we
take a hinge on the cost,
\begin{equation}
  R(\pi) \;=\; -\bigl( C(\pi) - c\,C^{*} \bigr)_{+},
  \qquad (x)_{+} = \max\{x,0\},
  \label{eq:hinge}
\end{equation}
which charges a tour only for the amount by which it exceeds the budget and so
vanishes on all of $\Omega_c$.\footnote{A penalty weight in front of the hinge is
redundant with $\alpha$, since the induced distribution depends only on their
ratio; we absorb it into $\alpha$.} These two ingredients already pin the
sampler down, and they pin it down to the optimal one.

\begin{proposition}[Entropy regularization with a flat reward is optimal]
\label{prop:gibbs}
The unique maximizer of Equation~\ref{eq:obj} over all distributions on a finite
solution space is the Gibbs distribution $p_\alpha(\pi) = e^{R(\pi)/\alpha}/Z$
with $Z = \sum_{\pi'} e^{R(\pi')/\alpha}$. Writing $\mu_\alpha$ for $p_\alpha$
conditioned on $\Omega_c$ and $R_{\min} = \min_{\pi \in \Omega_c} R(\pi)$,
\begin{equation}
  \Delta(\mu_\alpha)
  \;=\; \mathbb{E}_{\pi \sim \mathrm{unif}(\Omega_c)}
        \bigl[ e^{(R(\pi) - R_{\min})/\alpha} \bigr]
  \;\ge\; e^{(\bar R - R_{\min})/\alpha},
  \qquad \bar R = \mathbb{E}_{\mathrm{unif}(\Omega_c)}[R],
  \label{eq:flatness}
\end{equation}
so $\Delta(\mu_\alpha) = 1$ if and only if $R$ is constant on $\Omega_c$. In
particular the hinge of Equation~\ref{eq:hinge} gives $\Delta(\mu_\alpha) = 1$
for every $\alpha > 0$.
\end{proposition}

\begin{proof}
For the first claim,
$\mathbb{E}_p[R] + \alpha H(p) = \alpha \log Z - \alpha\,\mathrm{KL}(p \| p_\alpha)$;
the first term does not depend on $p$ and the second is non-negative, vanishing
only at $p = p_\alpha$. Conditioning on $\Omega_c$ replaces $Z$ by
$Z_c = \sum_{\pi \in \Omega_c} e^{R(\pi)/\alpha}$, so the lightest feasible tour
has mass $e^{R_{\min}/\alpha}/Z_c$ and Equation~\ref{eq:delta} gives
$\Delta(\mu_\alpha) = Z_c/(N_c e^{R_{\min}/\alpha})$, which is the stated
expectation; the lower bound is Jensen's inequality, with equality exactly when
$R$ is constant on $\Omega_c$. The hinge is zero throughout $\Omega_c$ by
Equation~\ref{eq:feasible}, hence constant there.
\end{proof}

Optimality is therefore a property of the reward rather than of the temperature.
What is required is that the reward be \emph{flat} on the feasible region, and
$\alpha$ then governs only how sharply the mass decays outside it, as
$\exp(-(C(\pi) - c\,C^{*})/\alpha)$. A reward that is not flat pays for it
exponentially, since Equation~\ref{eq:flatness} forces $\Delta$ to diverge as
$\alpha \to 0$ whenever the reward varies across $\Omega_c$. By
Remark~\ref{rem:chat} an inexact $\widehat{C}$ moves the boundary but leaves the
reward flat inside it, so $\Delta = 1$ survives the substitution.

\subsection{A matching lower bound}
\label{sec:lower}

Equation~\ref{eq:sample-size} prices non-uniformity at $\Delta$ in the worst
case over target sets. That worst case is taken over all $k$-subsets of
$\Omega_c$, not only over those that are optimal for some instance, so it is in
principle only an upper bound on what a sampler can be made to pay. We now show
that the price is real, by exhibiting instances on which a sampler with
non-uniformity $\Delta$ needs $\Omega(\Delta)$ times more samples than the
uniform one and returns a strictly worse set if it is not given them.

\begin{figure}[t]
\centering
\begin{tikzpicture}
% ---------- (a) the ring ----------
\begin{scope}
  \draw[gray!55, dashed] (0,0) circle (1.15);
  \foreach \a in {0,60,...,300} {
    \begin{scope}[rotate=\a]
      \draw[thick] (1.15,-0.30) -- (1.15,0.30);
      \fill (1.15,-0.30) circle (1.3pt);
      \fill (1.15,0)     circle (1.3pt);
      \fill (1.15,0.30)  circle (1.3pt);
      \fill (1.44,-0.08) circle (1.3pt);
    \end{scope}
  }
  \draw[rounded corners, gray!80] (1.00,-0.46) rectangle (1.62,0.46);
  \draw[gray!80] (1.62,0.46) -- (2.05,0.85);
  \node[font=\small, anchor=west] at (2.00,0.95) {cell $j$};
  \node[font=\small] at (0,-1.75) {(a) $g$ cells on a ring};
\end{scope}
% ---------- (b) traversal A ----------
\begin{scope}[xshift=4.6cm, yshift=-0.35cm]
  \draw[gray!50, dotted] (0,0) -- (2.4,0);
  \draw[gray!55, dotted] (-0.55,-0.45) -- (0,0);
  \draw[gray!55, dotted] (2.4,0) -- (2.95,-0.45);
  \draw[gray!60, dotted] (0.95,0.85) -- (0.95,0);
  \draw[thick] (0,0) -- (0.95,0.85);
  \draw[line width=1.7pt] (0.95,0.85) -- (1.2,0);
  \draw[thick] (1.2,0) -- (2.4,0);
  \fill (0,0) circle (1.6pt); \fill (1.2,0) circle (1.6pt);
  \fill (2.4,0) circle (1.6pt); \fill (0.95,0.85) circle (1.6pt);
  \node[font=\small, below] at (0,-0.05) {$p^1_j$};
  \node[font=\small, below] at (1.2,-0.05) {$p^2_j$};
  \node[font=\small, below] at (2.4,-0.05) {$p^3_j$};
  \node[font=\small, above] at (0.95,0.88) {$z_j$};
  \node[gray, font=\scriptsize] at (0.78,0.42) {$H$};
  \node[gray, font=\scriptsize] at (1.10,0.20) {$\vartheta$};
  \node[gray, font=\scriptsize] at (0.60,-0.42) {$L$};
  \node[gray, font=\scriptsize] at (1.80,-0.42) {$L$};
  \node[font=\small] at (1.2,-1.40) {(b) cheap traversal};
\end{scope}
% ---------- (c) traversal B ----------
\begin{scope}[xshift=8.6cm, yshift=-0.35cm]
  \draw[gray!50, dotted] (0,0) -- (2.4,0);
  \draw[gray!55, dotted] (-0.55,-0.45) -- (0,0);
  \draw[gray!55, dotted] (2.4,0) -- (2.95,-0.45);
  \draw[thick] (0,0) -- (1.2,0);
  \draw[line width=1.7pt] (1.2,0) -- (0.95,0.85);
  \draw[thick] (0.95,0.85) -- (2.4,0);
  \fill (0,0) circle (1.6pt); \fill (1.2,0) circle (1.6pt);
  \fill (2.4,0) circle (1.6pt); \fill (0.95,0.85) circle (1.6pt);
  \node[font=\small, below] at (0,-0.05) {$p^1_j$};
  \node[font=\small, below] at (1.2,-0.05) {$p^2_j$};
  \node[font=\small, below] at (2.4,-0.05) {$p^3_j$};
  \node[font=\small, above] at (0.95,0.88) {$z_j$};
  \node[font=\small] at (1.2,-1.40) {(c) costs $\eta$ more};
\end{scope}
\end{tikzpicture}
\caption{The lower-bound instance. (a) $g$ identical cells sit at equal angles
on a ring, and consecutive cells are joined by a single edge. (b), (c) A tour
enters a cell at $p^1_j$ and leaves at $p^3_j$, and has exactly two ways to
cover the four cities in between. The two share only the edge
$\{z_j, p^2_j\}$, drawn heavy, and the second costs $\eta$ more than the first.
The tangential offset $\vartheta$ of $z_j$ tunes $\eta$ continuously down to
zero while every other length is held fixed.}
\label{fig:gadget}
\end{figure}
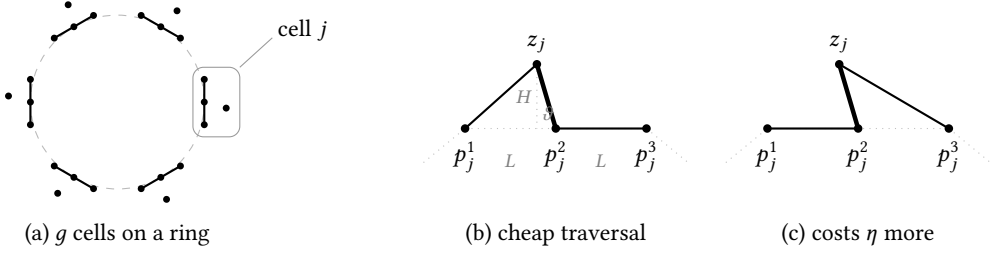

Figure~\ref{fig:gadget} shows the construction. The instance places $g$
identical cells at equal angles on a ring of radius $\rho$, joined in sequence
by one edge each. A cell contains four cities: $p^1_j$, $p^2_j$ and $p^3_j$ lie
evenly spaced along the tangent at spacing $L$, while $z_j$ sits a distance $H$
outside the ring and offset by $\vartheta$ along the tangent. A tour that enters
at $p^1_j$ and leaves at $p^3_j$ can cover the cell in exactly two ways, drawn
in panels (b) and (c). They share only the edge $\{z_j, p^2_j\}$, so they agree
on one of their three internal edges, and the second is more expensive by
\begin{equation}
  \eta \;=\; \sqrt{(L+\vartheta)^2 + H^2} - \sqrt{(L-\vartheta)^2 + H^2}
  \;=\; \Theta(\vartheta) .
  \label{eq:eta}
\end{equation}
The offset $\vartheta$ therefore tunes the cost gap continuously down to zero
while leaving every other length in the instance untouched, and it is the only
free parameter the argument needs.

Once $\vartheta$ is small enough, every tour within the budget has this shape,
so a feasible tour is described by which traversal it uses in each cell. Its
cost exceeds $C^{*}$ by $\eta$ for every cell of the expensive kind, and two
such tours overlap in $n - 2h$ edges, where $h$ counts the cells on which they
disagree. Choosing the tolerance so that $c\,C^{*} = C^{*} + \eta$ therefore
leaves $\Omega_c$ with exactly $g+1$ members: the cheap tour $\pi_0$ that takes
the cheap traversal everywhere, and the $g$ tours that deviate in a single cell.

Now take $k = g$. Deviating in different cells changes four edges, so the $g$
expensive tours are mutually at distance $2/(g+1)$, whereas each of them differs
from $\pi_0$ in two edges only and is therefore just $2/(2g+1)$ away from it.
The unique optimal set is thus the set of all $g$ expensive tours, and any other
choice of $g$ feasible tours must include $\pi_0$ and loses a factor
$(g+1)/(2g+1)$, which tends to $1/2$. This is what makes the family the right
test for Section~\ref{subsec:hit}. Asking the pool to contain a whole optimal
set is in general sufficient but not necessary, since the greedy rule needs only
\emph{some} $k$ mutually distant tours, and fixing one optimal set is a further
relaxation when the optimum is not unique; here the optimal set is unique and
must be hit exactly, so neither simplification costs anything. Nor does the
convention of asking for exact hits: by Remark~\ref{rem:resolution} no two
distinct tours lie within $4/(n+2)$ of each other, so allowing any smaller
radius is the same requirement.

\begin{theorem}[Lower bound]
\label{thm:lower}
For every $g \ge 3$ there is a Euclidean instance on $n = 4g$ cities and a
tolerance $c$ with $N_c = g+1$ such that, for $k = g$ and any
$\delta \in (0,1)$, the uniform sampler on $\Omega_c$ recovers the optimal set
with probability $1-\delta$ from $M \ge (g+1)\ln(g/\delta)$ samples, whereas any
sampler $\mu$ whose smallest mass falls on a tour of cost $c\,C^{*}$ needs more
than $(g+1)\Delta(\mu)/2$ samples to recover it with probability even $1/2$.
That is, every selection rule returns a set with
$D(\Pi) \le \tfrac{g+1}{2g+1} D^{*}_g(\Omega_c)$.
\end{theorem}

The proof is in Appendix~\ref{app:gadget}. What the theorem rules out is any reward that separates $\pi_0$ from the rest,
say by a gap $\Gamma > 0$. By Proposition~\ref{prop:gibbs} such a reward puts
the smallest mass on one of the expensive tours and gives
$\Delta(\mu_\alpha) = (e^{\Gamma/\alpha} + g)/(g+1)$, so $\Delta$ can be driven
as high as one likes by lowering $\alpha$. Sharpening the budget constraint in this way
therefore makes matters worse, while relaxing it flattens the reward but pushes
the policy towards the uniform distribution over \emph{all} tours, where almost
nothing is feasible. A reward that is already flat on $\Omega_c$, such as the
hinge of Equation~\ref{eq:hinge}, escapes the dilemma rather than balancing it:
its gap is zero, so $\Delta = 1$ at every temperature.

\subsection{The end-to-end guarantee}
\label{sec:endtoend}

Lemma~\ref{lem:sample} says when the pool contains an optimal set, and Theorem~\ref{thm:greedy} says what the selection rule does with such a pool. Composing them gives the guarantee in the form it is used, and the equal-cardinality property of Equation~\ref{eq:equal-cardinality} sharpens what
it says on the similarity scale at no algorithmic cost.

\begin{theorem}[End-to-end guarantee]
\label{thm:endtoend}
Fix $\delta \in (0,1)$ and an optimal set $\Pi^{*} \subseteq \Omega_c$, and
write $J^{*} = 1 - D^{*}_k(\Omega_c)$. If
\[
  M \;\ge\; \frac{\ln(k/\delta)}{q(\mu;\Pi^{*})},
\]
then with probability at least $1-\delta$ the set $\Pi$ returned by
Equation~\ref{eq:greedy} satisfies
\[
  D(\Pi) \;\ge\; \tfrac{1}{2}\,D^{*}_k(\Omega_c)
  \qquad\text{and}\qquad
  J(\Pi) \;\le\; \frac{1 + 3J^{*}}{3 + J^{*}} .
\]
Uniformly over target sets it suffices to take
$M \ge \Delta(\mu)\,N_c \ln(k/\delta)$, which for the uniform sampler on
$\Omega_c$ is $N_c \ln(k/\delta)$.
\end{theorem}

\begin{proof}
By Lemma~\ref{lem:sample}, with probability at least $1-\delta$ the pool
contains $\Pi^{*}$; since $\mathcal{P}_M \subseteq \Omega_c$ this forces
$D^{*}_k(\mathcal{P}_M) = D^{*}_k(\Omega_c)$, and in particular the pool holds
at least $k$ distinct tours. We work on that event.

The first bound is Theorem~\ref{thm:greedy} applied in the metric $d$. For the
second, recall from Remark~\ref{rem:hamming} that $(n-a)/n$ is a second metric
on tours and that both it and $d$ are strictly decreasing in the overlap $a$.
The score $\min_{q \in S} d(\pi,q)$ maximized at each step of
Equation~\ref{eq:greedy} is therefore a decreasing function of
$\max_{q \in S} a(\pi,q)$ under either metric, so the greedy rule selects the
same tours and the optimal $k$-subsets coincide. Writing $a_g$ and $a^{*}$ for
the largest overlap inside the returned and the optimal set,
Theorem~\ref{thm:greedy} in the second metric gives
$(n - a_g)/n \ge (n - a^{*})/(2n)$, that is, $a_g \le (n + a^{*})/2$.
Equation~\ref{eq:J-in-a} reads $a^{*}/n = 2J^{*}/(1+J^{*})$, so
\[
  \frac{a_g}{n} \;\le\; \frac{1}{2}\Bigl(1 + \frac{2J^{*}}{1+J^{*}}\Bigr)
  \;=\; \frac{1 + 3J^{*}}{2\,(1 + J^{*})} ,
\]
and substituting into $J = (a/n)\big/(2 - a/n)$, which is increasing in $a/n$,
gives the stated bound. The last claim is Proposition~\ref{prop:minimax}
together with Equation~\ref{eq:sample-size}.
\end{proof}

% =====================================================================
% Diverse TSP -- Section 5, revised
%
% Scope decisions:
%   - The grid / large-scale subsection has been removed in full.
%   - The network itself is not a contribution and now lives in
%     Appendix C. Only the MDP factorization stays in the main text.
%   - The reachability weight W_c is gone as a symbol. The quantity it
%     named appears once, priced in Delta, with no new notation.
%   - The edge-usage regularizer (former beta_2) has been removed.
%
% Notation:
%   - Cbar(T) is the mean cost of a traversal of T. Chat is reserved,
%     as in Section 3, for the upper bound on C*.
%   - r is the fixed root of the traversal; L is the number of
%     traversals used to estimate Cbar; B is the batch size.
% =====================================================================

\section{Sampling Trees and Converting Them to Tours}
\label{sec:model}

The sampler asked for in Section~\ref{sec:method} is uniform on $\Omega_c$, and
Proposition~\ref{prop:gibbs} says a flat reward with an entropy bonus produces
it. Neither statement says how to generate a tour, and generating one directly
is hard for a simple reason: the feasible set is very hard to hit. A tour drawn
at random costs $\Theta(n)$ against an optimum of $\Theta(\sqrt n)$, so a
tour-level policy has to place almost every edge well before a single sample
lands inside the budget and the hinge stops being constant. The entropy bonus
makes this worse rather than better, since it pulls the policy towards the
uniform distribution over \emph{all} tours, which is exactly where nothing is
feasible. Our own attempts to train a tour-level policy did not converge to a
usable model, which is consistent with this.

We therefore keep the objective and change the object being generated: the
policy samples a spanning tree, and a separate construction turns the tree into
a tour. Section~\ref{subsec:trees} illustrates our method, Section~\ref{subsec:loss} shows the learning objective, and Section~\ref{subsec:pushforward} analyzes what this method can and cannot guarantee compared with the uniform sampling on tours.

\subsection{Generating trees and converting them}
\label{subsec:trees}

We formulate the construction as a Markov decision process. Let
$\tau = (\tau_0,\dots,\tau_{\kappa-1})$ be the ordered list of distinct edges
selected during generation, with $\kappa = n-1$. At step $t$ the state is
$\tau_{0:t-1}$ and the action set contains every edge that does not close a
cycle, which depends only on the partial solution and is non-empty until step
$\kappa$. The policy factorizes as
\begin{equation}
p(\tau \mid G) = \prod_{t=0}^{\kappa-1}
  p_{\theta}\!\big(\tau_t \mid G, \tau_{0:t-1}\big),
\label{eq:autoregressive-policy}
\end{equation}
where each conditional puts mass only on the admissible edges. Any $n-1$
acyclic edges on $n$ vertices form a spanning tree, so the sequence always
terminates in a valid object. The network computing these conditionals follows
\citet{ma2019combinatorial} and is not a contribution of this work; we describe
it in Appendix~\ref{app:architecture}. We draw from the conditionals rather than
taking their $\arg\max$, since $\arg\max$ decoding is the zero-temperature limit
of Proposition~\ref{prop:gibbs} and collapses the pool onto a single solution.

Write $\Phi(T,\pi)$ for the probability that the construction applied to a tree
$T$ returns the tour $\pi$. We use two constructions.

\begin{itemize}
  \item \textbf{Method~1.} Run a depth-first search of $T$ from a fixed root
        $r$, choosing uniformly among the unvisited neighbours whenever a
        vertex has more than one, and take the order in which vertices are first
        visited; closing that order into a cycle gives $\pi$. This is the
        shortcut of the Euler circuit of the doubled tree.
  \item \textbf{Method~2.} Let $O$ be the set of odd-degree vertices of $T$, and
        add to $T$ a perfect matching on $O$ produced by a second policy of the
        same form.\footnote{The matching policy is the same network with
        $\kappa = |O|/2$ and a mask that admits an edge whenever both endpoints
        are still unmatched. Every statement below holds for it verbatim: the
        edges of a matching are pairwise disjoint, so every ordering of them is
        admissible and Proposition~\ref{prop:no-bias} applies unchanged.} The
        result has all degrees even, and shortcutting its Euler circuit gives
        $\pi$. This is the Christofides construction with a learned matching.
\end{itemize}

Both constructions are randomized, so $\Phi(T,\cdot)$ is in general a
distribution rather than a point mass. Figure~\ref{fig:flow} shows the resulting
pipeline.

\begin{figure}[htbp]
    \centering
    \includegraphics[width=0.95\linewidth]{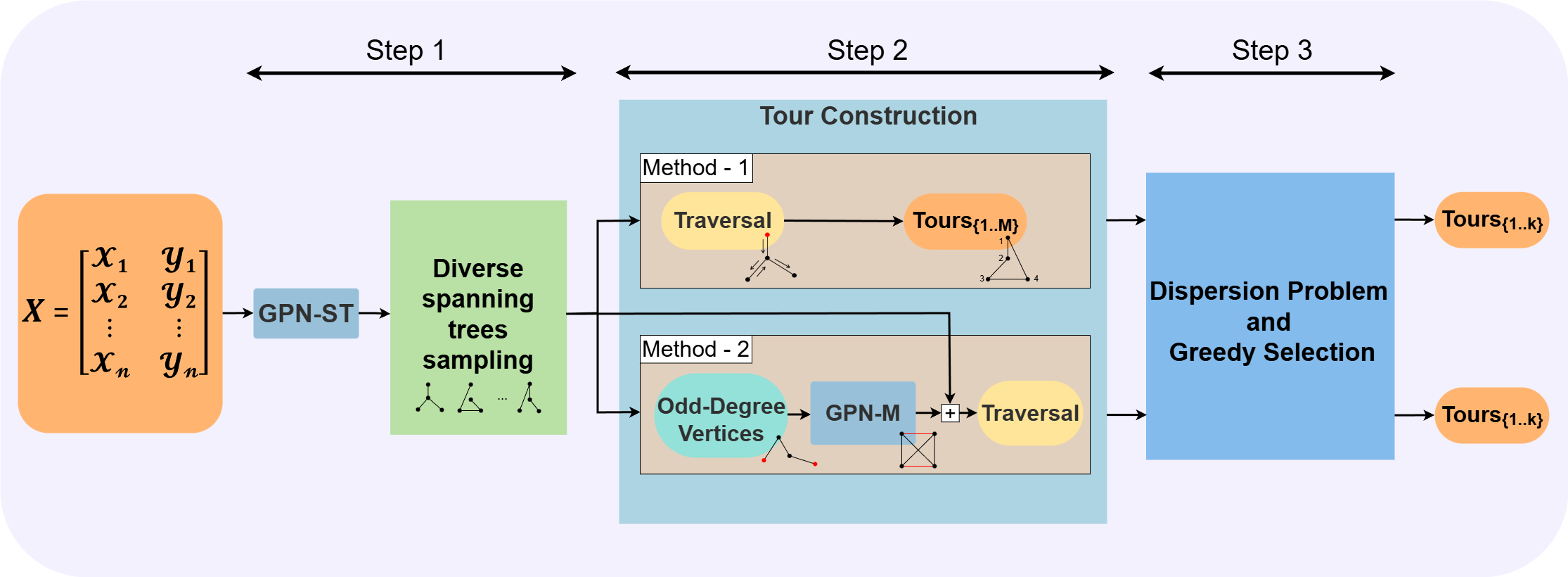}
    \caption{The pipeline. The network samples spanning trees; Method~1
    traverses a tree directly, while Method~2 first completes it with a learned
    matching on the odd-degree vertices; the resulting tours are filtered and
    passed to the selection stage.}
    \label{fig:flow}
\end{figure}

\subsection{Training objective}
\label{subsec:loss}

The policy is defined step by step, so its entropy accumulates step by step.
By the chain rule the per-step conditional entropies sum exactly to the entropy
of the sequence distribution, and we therefore define the entropy term as the
sum
\begin{equation}
  \mathcal{H}(\theta) \;=\; \sum_{t=1}^{\kappa}
    \mathcal{H}\bigl( p_\theta(\cdot \mid G, \tau_{0:t-1}) \bigr),
  \label{eq:entropy-term}
\end{equation}
rather than the average used by \citet{ma2019combinatorial}. With the sum, the
coefficient multiplying $\mathcal{H}(\theta)$ is exactly the temperature
$\alpha$ of Proposition~\ref{prop:gibbs}, so one variable governs both the
implementation and the analysis.

The reward is defined on tour cost, but the policy outputs a tree. Since the
construction is randomized, the cost of the resulting tour is a random variable
given $T$, and we estimate its mean with $L$ independent runs,
\begin{equation}
  \bar{C}(T) \;=\; \frac{1}{L}\sum_{l=1}^{L} C\bigl(\pi^{(l)}\bigr),
  \qquad \pi^{(l)} \sim \Phi(T,\cdot),
  \label{eq:tour-cost-est}
\end{equation}
and apply the hinge of Equation~\ref{eq:hinge} to it,
\begin{equation}
  R(T) \;=\; -\bigl( \bar{C}(T) - c\,\widehat{C} \bigr)_{+} .
  \label{eq:hinge-reward}
\end{equation}
The reward is thus flat on the set
\begin{equation}
  \Omega_c^{\mathrm{ST}} \;=\;
  \{\, T : \bar{C}(T) \le c\,\widehat{C} \,\}
  \label{eq:feasible-trees}
\end{equation}
of trees that meet the budget on average. For the baseline we take the reward of
the tree produced by greedy rollout, evaluated through the same hinge so that
the two are on one scale, and centre it over the batch:
$b_i = \widetilde{R}_i + \frac{1}{B}\sum_{j\le B}(R_j - \widetilde{R}_j)$, where
$\widetilde{R}_i$ is the greedy return for instance $i$. Putting the pieces
together, and writing $n_\pi$ for the number of times the tour $\pi$ occurs
among the $BL$ tours built in a batch, the policy gradient is
\begin{equation}
\begin{aligned}
\nabla_{\theta} J(\theta)
&= \frac{1}{B}\sum_{i=1}^{B}
   \Bigg[ \bigl( R_i - b_i \bigr)
   \sum_{t=1}^{\kappa} \nabla_{\theta} \log p_{\theta}\bigl( \tau_{i,t} \mid \tau_{i,0:t-1} \bigr)
   + \alpha \sum_{t=1}^{\kappa}
     \nabla_{\theta}\mathcal{H}\bigl( p_{\theta}(\cdot \mid \tau_{i,0:t-1}) \bigr)
   \Bigg] \\
&\qquad
   \;-\; \beta\,\nabla_\theta\, \mathbb{E}\bigl[ \log n_\pi \bigr] ,
\end{aligned}
\label{eq:policy-gradient}
\end{equation}
whose last term is explained in Section~\ref{subsec:pushforward}. Evaluating
Equation~\ref{eq:tour-cost-est} costs $O(nL)$ per sample, which is negligible
beside the $O(n^3)$ cost of generating one tree, so the overall complexity is
unchanged.

\subsection{(Un-)biased sampling from sequences to trees to tours}
\label{subsec:pushforward}

The objective is defined over tours while the policy puts mass on edge
sequences, so two mappings sit between them: sequences to trees, and trees to
tours. The first turns out to be harmless and the second does not.

\paragraph{Sequences to trees.}
Write $m(T)$ for the number of edge sequences that generate the tree $T$.

\begin{proposition}[Unbiased sampling of trees]
\label{prop:no-bias}
The entropy term of Equation~\ref{eq:entropy-term} is an unbiased estimator of
the entropy of the sequence distribution, and the maximizer of the objective
samples uniformly on $\Omega_c^{\mathrm{ST}}$.
\end{proposition}

\begin{proof}
The first claim is the chain rule for entropy applied to
Equation~\ref{eq:autoregressive-policy}. For the second, fix a tree $T$ and any
ordering of its edges. Every prefix is a subset of $T$, hence acyclic, hence
admissible, so all $(n-1)!$ orderings are valid sequences; and no ordering of a
non-tree edge set produces $T$. Therefore $m(T) = (n-1)!$ for every tree. Since
the reward depends on a sequence only through its tree, the maximizer assigns
equal probability to all $m(T)$ sequences generating $T$, and summing gives
$p_\theta(T) = m(T)e^{R(T)/\alpha}/Z$; the factor $m(T)$ is the same for every
tree, so it cancels against the normalizing constant and leaves the Gibbs
distribution over trees. By Equation~\ref{eq:hinge-reward} the reward is zero
throughout $\Omega_c^{\mathrm{ST}}$, so that distribution is uniform there.
\end{proof}

Outside $\Omega_c^{\mathrm{ST}}$ the mass decays as
$\exp(-(\bar{C}(T) - c\,\widehat{C})/\alpha)$. The constancy of $m(\cdot)$ is
what makes the argument work: had it varied across trees, the policy would have
been rewarded for favouring trees reachable in more ways. One caveat attaches to
the statement: it concerns the maximizer of the objective, which a parametric
policy reaches only approximately, and any non-uniformity among the sequences
generating one tree is a separate approximation error that $m(T)$ does not
measure.

\paragraph{Trees to tours.}
Uniformity on trees is not uniformity on tours: one tree produces many tours and
many trees produce the same tour, so pushing the tree law through $\Phi$ gives
$p_\theta(\pi) = \sum_T p_\theta(T)\Phi(T,\pi)$ and the weights do not survive
intact. Two questions follow. Can the relaxation represent the sampler we want
at all, and what does the sampler it actually produces look like? The answers
are yes and no, and the gap between them is where the rest of the paper acts.

\begin{proposition}[The relaxation still achieves uniform sampling]
\label{prop:reachable}
For every distribution $\nu$ on $\Omega_c$ there is a policy of the form of
Equation~\ref{eq:autoregressive-policy} whose tree distribution is supported in
$\Omega_c^{\mathrm{ST}}$ and whose pushforward through either construction is
exactly $\nu$.
\end{proposition}

\begin{proof}
Let $\mathcal{Q}$ collect the trees obtained by deleting from a tour of
$\Omega_c$ one of its two edges at the root $r$. Such a tree is a Hamiltonian
path with $r$ as an endpoint, hence a spanning tree. Under Method~1 the search
starts at $r$, which has one unvisited neighbour, and after each step the
current vertex again has exactly one, so no branch choice ever arises: the visit
order is the path order and closing it returns the tour $\pi$ it came from.
Under Method~2 the odd-degree vertices of the path are its two endpoints, which
are the endpoints of the deleted edge, so the matching on them is that edge
alone; adding it back gives $\pi$, already a cycle and left unchanged by
shortcutting. Either way the construction is deterministic, so
$\Phi(T,\cdot) = \delta_\pi$ and every run returns cost $C(\pi)$, giving
$\bar{C}(T) = C(\pi)$ and hence $\mathcal{Q} \subseteq \Omega_c^{\mathrm{ST}}$.
Each tour of $\Omega_c$ contributes exactly two trees to $\mathcal{Q}$, so
placing mass $\nu(\pi)/2$ on each of the two trees that close into $\pi$ gives a
distribution on $\mathcal{Q}$ that pushes forward to $\nu$. Finally, every
ordering of a tree's edges is admissible and any joint distribution over
sequences factorizes as in Equation~\ref{eq:autoregressive-policy}, so this
distribution over trees, like any other, is induced by some policy in the class.
\end{proof}

By Proposition~\ref{prop:reachable}, the optimal sampler is within reach of this architecture. However, we lack the  objective that singles it out. The reward we have reads only tour cost, and every tree in $\Omega_c^{\mathrm{ST}}$ scores the same, so it cannot tell apart the trees that traverse to a single tour from those that spread over many. The exponential factor therefore leaves the sum, and the law we actually induce is
\begin{equation}
  p_\theta(\pi) \;\propto\; \sum_{T \in \Omega_c^{\mathrm{ST}}} \Phi(T,\pi) ,
  \label{eq:induced-tour-law}
\end{equation}
up to the exponentially suppressed contribution of trees outside the budget.
Restricting to $\Omega_c$ leaves these relative weights unchanged, so by
Equation~\ref{eq:delta} we pay a $\Delta$ equal to the ratio of the average of
the right-hand side to its smallest value, which exceeds one as soon as some
tours are reachable by more trees than others.

We can at least measure this. Because no reward term survives in
Equation~\ref{eq:induced-tour-law}, the number of times a tour recurs within a
batch estimates its weight there up to a constant, and the last term of
Equation~\ref{eq:policy-gradient} pushes down on the tours that recur most, at a
cost of $O(BLn)$ per batch. This treats the symptom rather than the cause and
does not cancel the ratio, so we present it as a heuristic. 
\section{The Conflict Graph}
\label{sec:conflict}

Section~\ref{sec:method} asked the pool to contain a whole optimal set. However, it may not be necessary in many cases. The reason is that, Equation~\ref{eq:greedy} never looks for a
particular set of tours; it needs only that the pool hold \emph{some} $k$ tours that lie far apart. The two requirements agree when the best separation is achieved by essentially one set of tours, as on the instances of Section~\ref{sec:lower}, where the optimal set is unique and has to be hit exactly. This section replaces the coverage requirement by the weaker one. Doing so isolates a single scalar through which the sampler enters the guarantee and in the end of Section~\ref{subsec:objective} we would show that the corresponding objective function has the monotone property, which make it suitable for trainable algorithm.

\subsection{The independent set v.s. D-TSP}
\label{subsec:identity}

\begin{definition}[Conflict graph]
\label{def:conflict}
Let $\rho > 0$ and let $\pi_1,\dots,\pi_N$ be a finite list of tours, possibly
with repetitions. The conflict graph $G_\rho(\pi_1,\dots,\pi_N)$ has vertex set
$\{1,\dots,N\}$ and joins $i \ne j$ whenever $d(\pi_i,\pi_j) < \rho$. We write
$G_\rho(S)$ when the list is a set. An independent set is a collection of
indices carrying pairwise distinct tours at distance $\rho$ or more, and
$\mathrm{ind}(H)$ denotes the size of a largest one.
\end{definition}

An independent set of $G_\rho(S)$ is a set of tours that are pairwise at
distance $\rho$ or more, which is exactly what D-TSP asks for. The two problems
are therefore the same problem read in two directions.

\begin{lemma}[Dispersion as an independence number]
\label{lem:identity}
For any finite $S$ with $|S| \ge k$,
\[
  D^{*}_k(S) \ge \rho
  \qquad\Longleftrightarrow\qquad
  \mathrm{ind}\bigl(G_\rho(S)\bigr) \ge k .
\]
In particular
$\mathrm{ind}(G_\rho(\Omega_c)) = \max\{\, k : D^{*}_k(\Omega_c) \ge \rho \,\}$.
\end{lemma}

\begin{proof}
A subset of $S$ is independent in $G_\rho(S)$ precisely when all its pairwise
distances are at least $\rho$, that is, when $D(\Pi) \ge \rho$. Such a subset of
size $k$ exists if and only if $D^{*}_k(S) \ge \rho$ by
Equation~\ref{eq:objective}. The second statement follows because $D^{*}_k(S)$
is non-increasing in $k$.
\end{proof}

The equivalence is exact rather than a relaxation, and it inverts the objective:
$D^{*}_k$ reports the separation available at a given budget of tours, while
$\mathrm{ind}(G_\rho(\Omega_c))$ reports how many tours are available at a given
separation. 

Now, it is clear that the sampled graph is an induced subgraph of the instance graph $G_\rho(\Omega_c)$ and
\begin{equation}
  \mathrm{ind}\bigl(G_\rho(\mathcal{P}_M)\bigr)
  \;\le\;
  \mathrm{ind}\bigl(G_\rho(\Omega_c)\bigr),
  \label{eq:ind-monotone}
\end{equation}
with the sampled tours $\mathcal{P}_M$, where $\mathcal{P}_M \subseteq \Omega_c$. Composing Lemma~\ref{lem:identity} with Theorem~\ref{thm:greedy} closes the
loop. If $\mathrm{ind}(G_\rho(\mathcal{P}_M)) \ge k$, then
$D^{*}_k(\mathcal{P}_M) \ge \rho$, and the set $\Pi$ returned by
Equation~\ref{eq:greedy} on the pool satisfies $D(\Pi) \ge \rho/2$. The
selection stage is therefore untouched, and the demand on the sampler is now
explicit: make the independence number of the pool reach $k$ at as large a
$\rho$ as possible. The largest $\rho$ the pool sustains is exactly $D^{*}_k(\mathcal{P}_M)$, so a
better sampler is one that raises $D^{*}_k(\mathcal{P}_M)$. That is, assign probability on each vertex such that the $\rho$ value can be as high as possible and the sampled graph has $\mathrm{ind}(G_\rho(\mathcal{P}_M)) \ge k$ with high probability.

\subsection{Analysis on \texorpdfstring{$\rho$}{rho}}
\label{subsec:collision}
To analyze what is the best sampler that can maximize $\rho$ such that the $\mathrm{ind}(G_\rho(\mathcal{P}_M))$ can reach $k$ with high probability, let $\mu$ be the sampler conditioned on $\Omega_c$, as in
Section~\ref{sec:method}, and define the \emph{collision rate}
\begin{equation}
  p_\mu(\rho) \;=\; \Pr\bigl[\, d(X,Y) < \rho \,\bigr],
  \qquad X, Y \sim \mu \ \text{independently}.
  \label{eq:collision}
\end{equation}
The event $X = Y$ is included, since $d(X,X) = 0$. In the following we show that $p_\mu(\rho)$ can bound the quality of the sampler on its own.

\begin{proposition}[Caro–Wei Sampling bound]
\label{prop:carowei}
Draw $M$ tours i.i.d.\ from $\mu$. Then
\begin{equation}
  \mathbb{E}\bigl[\mathrm{ind}(G_\rho(\mathcal{P}_M))\bigr]
  \;\ge\; \frac{M}{1 + (M-1)\,p_\mu(\rho)} ,
  \label{eq:carowei}
\end{equation}
which increases in $M$ to $1/p_\mu(\rho)$. Moreover, with probability at least
$1-\delta$, $\mathrm{ind}(G_\rho(\mathcal{P}_M))$ exceeds the right-hand side of
Equation~\ref{eq:carowei} minus $\sqrt{(M/2)\ln(1/\delta)}$.
\end{proposition}

\begin{proof}
The Caro--Wei bound gives
$\mathrm{ind}(H) \ge \sum_{v} (1 + \deg v)^{-1}$
\citep{caro1979new, wei1981lower}. Each vertex of $G_\rho(\mathcal{P}_M)$ has
$\mathbb{E}[\deg] = (M-1)p_\mu(\rho)$ by Equation~\ref{eq:collision}, and
$x \mapsto (1+x)^{-1}$ is convex, so taking expectations and applying Jensen's
inequality\footnote{Jensen's inequality states that
$\mathbb{E}[f(Z)] \ge f(\mathbb{E}[Z])$ for a convex $f$.} to each of the $M$ terms gives Equation~\ref{eq:carowei}. For the
second claim, replacing one of the $M$ samples changes
$\mathrm{ind}(G_\rho(\mathcal{P}_M))$ by at most one, so the bounded-difference
inequality \citep{mcdiarmid1989method} applies with unit increments.
\end{proof}

\begin{remark}[The expectation form]
\label{rem:expectation-form}
For any $\beta > 1$, a sampler with $p_\mu(\rho) \le 1/(\beta k)$ satisfies
\[
  \mathbb{E}\bigl[\mathrm{ind}(G_\rho(\mathcal{P}_M))\bigr] \;\ge\; k
  \qquad\text{as soon as}\qquad
  M \;\ge\; \frac{\beta k - 1}{\beta - 1}.
\]
The largest usable threshold is $\rho = D^{*}_{\beta k}(\Omega_c)$: the benchmark is the best $\beta k$ tours rather than the best $k$.
\end{remark}

Lemma~\ref{lem:identity} converts an independent set into a dispersion, and Proposition~\ref{prop:carowei} says how many samples produce one. Composing them with Theorem~\ref{thm:greedy} have the following guarantee.

\begin{theorem}[End-to-end guarantee from the collision rate]
\label{thm:endtoend-conflict}
Given $\rho$, where $ D^{*}_{4k}(\Omega_c) \geq \rho > 0$ and $\delta \in (0,1)$ with $\ln(1/\delta) \le k/2$. If the
sampler satisfies
\[
  p_\mu(\rho) \;\le\; \frac{1}{4k},
\]
then drawing $M = 4k$ tours gives, with probability at least $1-\delta$,
\[
  D(\Pi) \;\ge\; \tfrac{1}{2}\,\rho .
\]
\end{theorem}

\begin{proof}
With $M = 4k$ and $p_\mu(\rho) \le 1/(4k)$ we have $(M-1)p_\mu(\rho) < 1$, so
Equation~\ref{eq:carowei} gives
$\mathbb{E}[\mathrm{ind}(G_\rho(\mathcal{P}_M))] > M/2 = 2k$. Replacing one of
the $M$ samples changes the independence number by at most one, so the
bounded-difference inequality gives
$\Pr[\mathrm{ind} \le \mathbb{E}[\mathrm{ind}] - k] \le e^{-2k^2/M} =
e^{-k/2} \le \delta$. On the complementary event
$\mathrm{ind}(G_\rho(\mathcal{P}_M)) \ge k$, so
$D^{*}_k(\mathcal{P}_M) \ge \rho$ by Lemma~\ref{lem:identity}, and
Theorem~\ref{thm:greedy} returns a set with $D(\Pi) \ge \rho/2$.
\end{proof}

Comparing the two guarantees isolates what the change of criterion buys. Theorem~\ref{thm:endtoend} certifies $D(\Pi) \ge \tfrac{1}{2} D^{*}_k(\Omega_c)$ but needs $\Delta(\mu)\,N_c \ln(k/\delta)$ samples, which is exponential in $n$ and, by Theorem~\ref{thm:lower}, cannot be reduced. Theorem~\ref{thm:endtoend-conflict} needs $4k$, linear in $k$ and independent of $N_c$ and of $\delta$, but certifies $D(\Pi) \ge \tfrac{1}{2} D^{*}_{4k}(\Omega_c)$: the benchmark weakens from the best $k$ tours to the best $4k$, since the pool must now be shown to hold a surplus of separated tours rather than one specific set. The price of a polynomial sample size is therefore paid in the benchmark, not in the approximation factor, which stays at two in both.

%%%% mention finding rho close to upperbound is NP-hard
Proposition~\ref{prop:carowei} gives that the quality of a sampler is bounded by $p_\mu(\rho)$. One interesting fact is that finding the minimum value of $p_\mu(\rho)$ is actually NP-hard. See the following.

\begin{proposition}
\label{prop:motzkin}
$\displaystyle \min_{\mu} p_\mu(\rho) = \frac{1}{\mathrm{ind}(G_\rho(\Omega_c))}$,
is attained at the uniform distribution on any maximum independent set of $G_\rho(\Omega_c)$.
\end{proposition}

\begin{proof}
Writing $A_\rho$ for the adjacency matrix of $G_\rho(\Omega_c)$ and treating
$\mu$ as a vector on the simplex, Equation~\ref{eq:collision} reads
$p_\mu(\rho) = \mu^{\!\top}(A_\rho + I)\mu$, the identity accounting for the
event $X = Y$. The claim is the Motzkin--Straus theorem in its form for the
independence number \citep{motzkin1965maxima}.
\end{proof}

Propsition~\ref{prop:motzkin} shows that given a value of $\rho$, if we can have the optimal sampler with minimized $p_\mu(\rho)$, we can directly have $\mathrm{ind}(G_\rho(\Omega_c))$, which is the answer of the inverted version of D-TSP. Thus, the actual challenge of learning a good sampler is to address how to close to the bound of $\rho$ and minimizing $p_\mu(\rho)$ simultaneously, which motivates the learning objective described in the next subsection.

%Lastly, to see how much benefit we can gain compared with the analysis in Proposition~\ref{prop:minimax}, we can rearrange Equation~\ref{eq:carowei} which gives the sample size directly: if
%$p_\mu(\rho) < 1/k$, then
%\begin{equation}
  %M \;\ge\; \frac{k\bigl(1 - p_\mu(\rho)\bigr)}{1 - k\,p_\mu(\rho)}
  %\label{eq:sample-conflict}
%\end{equation}
%suffices for $\mathbb{E}[\mathrm{ind}(G_\rho(\mathcal{P}_M))] \ge k$, and the concentration statement of Proposition~\ref{prop:carowei} turns this into a high-probability guarantee at a given $M$, at the cost of an additive $\sqrt{(M/2)\ln(1/\delta)}$. The requirement is polynomial in $k$, in contrast with the $\Delta \cdot N_c \ln(k/\delta)$ of Equation~\ref{eq:sample-size}, which is exponential in $n$. This advantage holds only if $\mathrm{ind}\bigl(G_\rho(\Omega_c)\bigr) > k$.

\subsection{The learning objective}
\label{subsec:objective}

Proposition~\ref{prop:carowei} tells the sampler to minimize $p_\mu(\rho)$, and
$\mu$ is the policy conditioned on $\Omega_c$, so the hinge of
Equation~\ref{eq:hinge} and the collision rate act on different things: the
hinge decides which tours are admissible and is flat once inside, while the
collision rate shapes the distribution within. The two still pull against each
other, since spreading mass reduces collisions most cheaply by leaving the
feasible region, and $\alpha$ sets the exchange rate as before.

Compared with the analysis in Section~\ref{subsec:hit}, the criterion is not an addition to the entropy bonus but a replacement for it.
As $\rho \to 0$ no two distinct tours are within the threshold, by
Remark~\ref{rem:resolution}, and
\begin{equation}
  p_\mu(0^{+}) \;=\; \sum_{\pi \in \Omega_c} \mu(\pi)^2 ,
  \label{eq:collision-zero}
\end{equation}
whose minimizer over the simplex is the uniform distribution on $\Omega_c$. The
two sections therefore agree at $\rho = 0$ on which sampler is best, though they
reach it through different functionals.

Fortunately, the variable~$\rho$ need not be tuned. Writing
$g(\eta) = \eta - k\,\mathbb{E}[(\eta - d(X,Y))^{+}]$ for $X,Y \sim \mu_\theta$
independent, we have $g'(\eta) = 1 - k\,p_\mu(\eta)$, so $g$ is concave and its
maximum is attained where $p_\mu(\eta) = 1/k$. This is the
Rockafellar--Uryasev representation of the conditional value at risk
\citep{rockafellar2000optimization}: at the optimum, $g(\eta)$ is the mean of
the smallest $1/k$ fraction of pairwise distances, and raising it is what the
sampler should do.

Collecting the pieces, the sampler is trained by
\begin{equation}
  \max_{\theta} \; \max_{\eta} \;\;
  \mathbb{E}_{\pi \sim p_\theta}\bigl[ R(\pi) \bigr]
  \;+\; \alpha \Bigl\{\, \eta - k\,\mathbb{E}_{X,Y \sim \mu_\theta}
  \bigl[ (\eta - d(X,Y))^{+} \bigr] \Bigr\} ,
  \label{eq:final-objective}
\end{equation}
with $R$ the hinge of Equation~\ref{eq:hinge}. The two expectations are taken
under different laws: the first under the policy itself, since the hinge is zero
inside $\Omega_c$ and carries gradient only on infeasible samples, and the
second under the policy conditioned on $\Omega_c$, which in a batch means the
samples that pass the filter. Equation~\ref{eq:final-objective} replaces the
entropy bonus of Equation~\ref{eq:obj} rather than adding to it, and $\alpha$
plays the same role as before. The duplicate penalty of
Section~\ref{sec:model} is likewise subsumed, being the $\rho \to 0$ case of the
second term. Both maximizations can be carried out by gradient ascent, the
inner one because $g$ is concave and $(\cdot)^{+}$ is differentiable almost
everywhere. The threshold is thus read off rather than set, with one caveat:
$\eta$ tracks the $1/k$ quantile of the distances the current policy produces,
which is an estimate of what is presently attainable and not of
$D^{*}_k(\Omega_c)$. The objective pushes it upwards. Notice that the bound of Equation~\ref{eq:carowei} is monotone in $p_\mu(\rho)$, so any decrease in the collision rate is a certified improvement
of the guarantee, whether or not the true independence number of the sampled graph rises with it.

\section{Experiments}
\label{sec:Results}
 
The experiments answer two separate questions. The first is whether the analysis of Sections~\ref{sec:method} and~\ref{sec:conflict} describes what actually happens: whether the coverage requirement is as demanding as
Section~\ref{sec:method} claims, whether the collision rate really governs the
dispersion a pool can support, and whether the sample size saturates where
Proposition~\ref{prop:carowei} says it should. The second is how the resulting
sampler compares with existing methods, on both in-distribution and
out-of-distribution instances.
 
\subsection{Setup}
\label{subsec:setup}
 
\textbf{Datasets.} We evaluate on two groups of instances. The first consists
of randomly generated graphs with 40 cities, matching the scale used during
training. The second consists of all TSPLIB~\citep{reinelt1991tsplib} instances
with 51 to 101 cities; the sampler is trained only on instances from the first
group and is applied to these without adaptation. Consistent with our training
protocol, the coordinates of all cities across all evaluation datasets are
independently normalized to $[0,1]^2$ prior to inference to prevent spatial
distortion.
 
Two disjoint sets of 20 random 40-city instances are used: one for the
validation experiments of Section~\ref{subsec:validation} and one for the
comparison of Section~\ref{subsec:headtohead}. The latter is the set on which
the baseline methods were originally run, so that their numbers can be reused
without modification.
 
\textbf{Competitors.} We benchmark our framework against representative
baselines from both the traditional and the neural categories:
\begin{itemize}
  \item \textbf{Traditional heuristics.} We include the Niching Memetic
  Algorithm (NMA)~\citep{huang2019niching}, widely regarded as the
  state-of-the-art heuristic for the Multi-Solution TSP. We also compare
  against the Bi-Criteria Heuristic~\citep{gao2022obtaining}. To allow a
  direct comparison, we adapt its output---diverse spanning trees---into TSP
  tours using the same traversal employed in our own tree-to-tour construction.
  \item \textbf{Neural baseline.} We evaluate against
  RF-MA3S~\citep{li2025diversity}, a DRL-based solver published in KDD'25 and
  designed specifically for diversity optimization. It employs an
  encoder--decoder architecture with adaptive active search to balance solution
  quality and diversity.
\end{itemize}
 
\textbf{Implementation.} All experiments are conducted on a machine equipped
with an NVIDIA RTX 4090 GPU and an Intel Core i7-13700K CPU. The traditional
heuristics are executed on the same hardware to ensure a fair comparison.
 
\textbf{Our setting.} A single sampler~\footnote{The Method 1 and Method 2 performs similar qualities among all the instances. Here we only report Method 1 to avoid redundancy.} is trained once and used for both
experiments, on 40 cities drawn uniformly from $[0,1]^2$, with the hinge reward
of Equation~\ref{eq:hinge} at $c = 2.0$. Training runs for 100 epochs with
1{,}000 steps per epoch and a batch size of 256, optimized with Adam at a
constant learning rate of $5 \times 10^{-4}$; validation at the end of each
epoch monitors convergence. Training takes approximately 12.7 hours. The
entropy coefficient is $\alpha = 1.3 \times 10^{-3}$ in the summed form of
Equation~\ref{eq:entropy-term}. The same checkpoint is used for every
experiment in this section, with no fine-tuning of any kind.
 
\textbf{Baseline settings.} To ensure a fair comparison, we align the
hyperparameters of the baseline methods with the recommendations in their
respective original papers.
\begin{itemize}
  \item For the Bi-Criteria Heuristic we adopt the methodology
  of~\citet{gao2022obtaining}, which formulates the search for diverse spanning
  trees as a constrained minimum spanning tree problem, where finding a
  distinct tree is treated as minimizing edge usage frequency subject to a cost
  budget. We implement their Lagrangian relaxation-based approximation and
  generate 20 trees for each relaxation factor in $[1.1, 3.0]$, converting each
  tree into a tour by the same traversal used for our own.
  \item For the evolutionary baseline, NMA~\citep{huang2019niching}, we
  execute the algorithm with 8 distinct population sizes
  $N_{pop} \in \{150, 200, \dots, 500\}$, running it repeatedly at each setting
  to accumulate 125 solutions. We use a mutation rate $P_m = 0.1$, a stagnation
  threshold of 60 for critical edge updates, a niche radius refinement
  parameter of 2, and an adaptive neighbourhood constrained to
  $M_{\min} = 4$ and $M_{\max} = 6$.
  \item For the neural baseline, RF-MA3S~\citep{li2025diversity}, we train the
  model on instances with 40 cities, strictly following the hyperparameter
  settings recommended in their original paper. Training takes approximately 11
  hours.
\end{itemize}
 
\textbf{Inference and selection.} Unlike methods that rely on computationally
expensive active search at inference time~\citep{li2025diversity}, our approach
uses the stochasticity of the trained policy directly. We draw a pool of
1{,}000 tours, discard those whose cost exceeds $c\,\widehat{C}$, and apply
Equation~\ref{eq:greedy} to the remainder. Because sampling is stochastic even
for a fixed checkpoint, every reported number is the median over three
independent seeds, with the seed range given where it is non-zero. If fewer
than $k$ feasible tours survive the filter, we record the setting as not
achieved rather than resampling.
 
\textbf{Budget reference.} The budget reference $\widehat{C}$ is obtained from
Google OR-Tools on the normalized coordinates, for both the random and the
TSPLIB instances. On the six TSPLIB instances whose optimum is published, the
ratio $\widehat{C}/C^{*}$ lies between $1.0002$ and $1.0169$, so the budget is
in effect at most $1.7\%$ looser than the nominal $c$. By
Remark~\ref{rem:chat} this does not affect any statement in the analysis.
 
\textbf{Metrics.} We report $D(\Pi)$, the smallest pairwise Jaccard distance
inside the selected set, which is the quantity Definition~\ref{def:dtsp}
maximizes and the scale on which Theorem~\ref{thm:greedy} is multiplicative;
higher is better. The discussion occasionally reports the largest pairwise
Jaccard \emph{similarity} $J(\Pi) = 1 - D(\Pi)$ where it aids comparison, but
note that the same result looks very different on the two scales, since the
guarantee is multiplicative in $d$ and only additive in $J$. All evaluations
use budget factors $c \in \{2.0, 2.5, 3.0, 4.0\}$ with
$k \in \{80, 120, 160, 320\}$ respectively, $k$ growing as the feasible region
$\Omega_c$ expands.
 
\subsection{Validating the theoretical analysis}
\label{subsec:validation}
 
\begin{figure}
    \centering
    \includegraphics[width=0.9\linewidth]{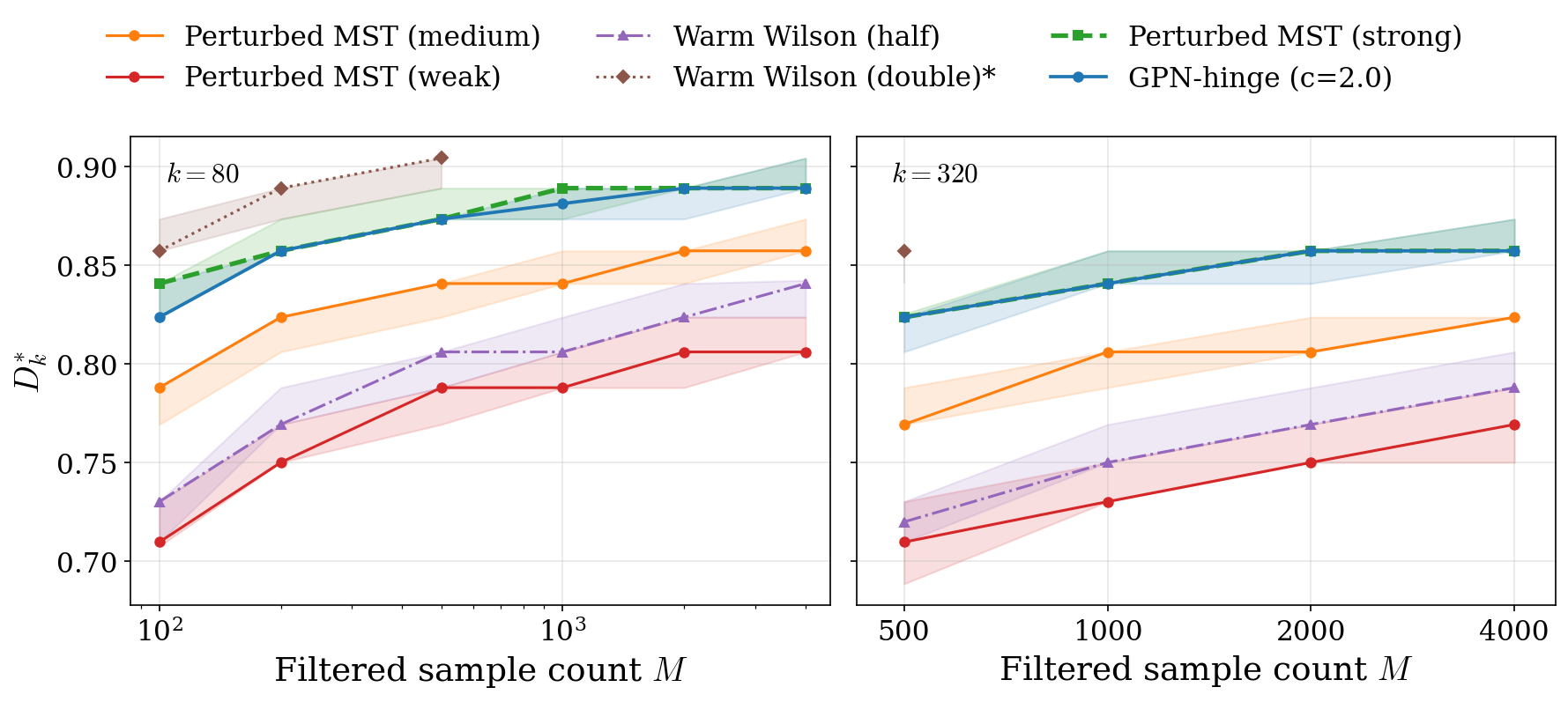}
    \caption{$D^{*}_k(S)$ versus the number of samples, across the samplers
    considered.}
    \label{fig:e3eff}
\end{figure}
 
Recall that Section~\ref{sec:method} derives a quality guarantee from an
$\varepsilon$-cover argument. These experiments ask whether the sample size it
demands is actually necessary in a common case, namely instances drawn from
uniformly generated graphs. Ideally, far fewer samples should be needed, and
the number should be inversely proportional to the collision rate
$\hat{p}_\mu$ (Equation~\ref{eq:collision}), as long as the sampler is close to
the sampling distribution that D-TSP calls for.
 
To examine this, we use seven pool-generation procedures: the learned sampler
(GPN-hinge, our method), perturbed minimum spanning trees at three perturbation
strengths,\footnote{The perturbation strength is set from the edge weight
distribution: weak (10th quantile), medium (25th quantile) and strong (50th
quantile).} and three warm-started weighted Wilson
samplers~\citep{wilson1996generating}---calibrated, half and
double.\footnote{Weighted Wilson sampling runs Wilson's loop-erased random walk
with transition probabilities proportional to $\exp(-w_{uv}/\alpha)$, which
draws a spanning tree $T$ with probability proportional to
$\exp(-w(T)/\alpha)$. The temperature $\alpha$ is set per instance by bisection
so that the mean normalized tour cost matches that of the learned sampler; the
two uncalibrated variants use $0.5\alpha$ and $2\alpha$.} Their role is to
provide variation across sampling methods rather than to compete with one
another. All results use the 20 random 40-city instances reserved for this
purpose, with $M \approx 4{,}000$ and at most 5{,}000 raw draws.
 
\begin{table}[t]
\caption{Sampler behaviour at $k=320$ under a budget of roughly 5{,}000 raw
draws. Acceptance and raw draws are medians over the 20 instances; the sample
count $M$ refers to tours retained after the budget filter.}
\label{tab:sat}
\centering
\begin{threeparttable}
\begin{tabular}{lrrrr}
\toprule
Sampler & Acceptance & $M$ & $D^{*}_{320}$ & Raw draws (median) \\
\midrule
Learned sampler         & 0.993 & 2000             & $0.8571^{\dagger}$ & 2014 \\
Perturbed MST (strong)  & 0.986 & 2000             & $0.8571^{\dagger}$ & 2027 \\
Perturbed MST (medium)  & 1.000 & 4000             & 0.8235             & 4000 \\
Perturbed MST (weak)    & 1.000 & 4000             & 0.7692             & 4000 \\
\midrule
Warm Wilson (double)    & 0.154 & $500^{\ddagger}$ & 0.8571             & 3217 \\
\bottomrule
\end{tabular}
\begin{tablenotes}\footnotesize
\item[$\dagger$] The value is unchanged beyond 2{,}000 samples.
\item[$\ddagger$] Warm Wilson (double) is reported at $M=500$. Beyond that
point the number of instances with a pool of the required size drops to 11,
and falls further as $M$ grows.
\end{tablenotes}
\end{threeparttable}
\end{table}
 
Figure~\ref{fig:e3eff} plots the pairwise Jaccard distance at $k = 80$ and
$k = 320$ against the number of samples; the calibrated Wilson sampler traces
exactly the same curve as GPN-hinge, so only one of the two is shown. Three
things are visible. First, an $\varepsilon$-cover is not necessary at the scale
of these instances: 250 samples already suffice for Equation~\ref{eq:greedy} to
reach $90\%$ of the best dispersion the full pool supports, for every sampler
reported. Second, among all the samplers, warm-started Wilson (double), the
grey curve, attains the highest Jaccard distance at $M = 500$. Its acceptance
rate is low, however, so it needs more raw draws than the others across the 20
instances, as Table~\ref{tab:sat} shows. This suggests that the hinge reward is
indeed what keeps a sampler inside the intended range of tour cost. Finally,
the number of samples eventually saturates the achievable Jaccard distance. At
$k = 80$ every curve flattens, most of them between five hundred and two
thousand samples, after which additional samples buy nothing. This is the
ceiling $1/p_\mu(\rho)$ implied by Proposition~\ref{prop:carowei}: beyond it
the limitation is the distribution, not the number of draws.
 
\begin{figure}
    \centering
    \includegraphics[width=0.7\linewidth]{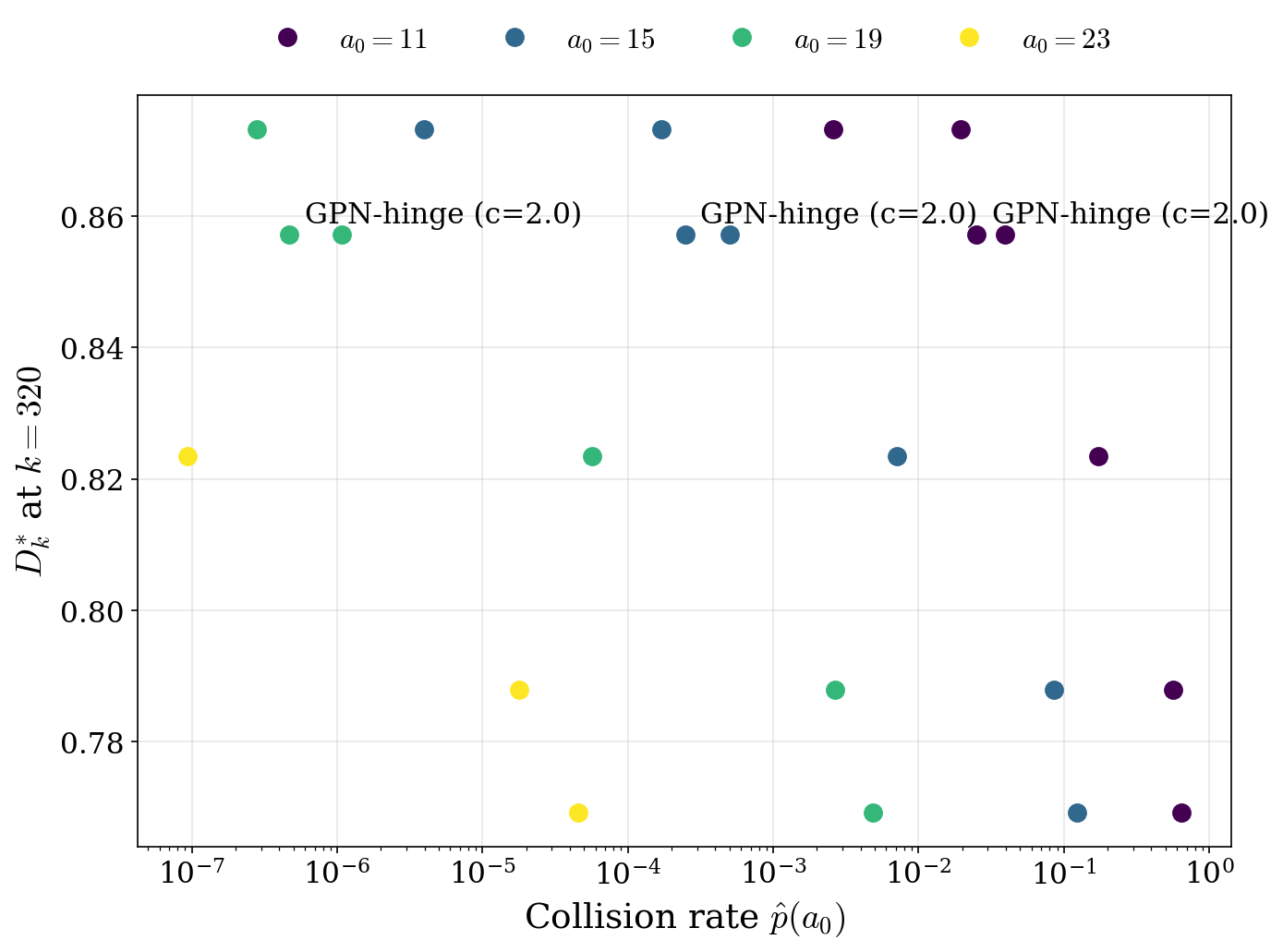}
    \caption{Collision rate $\hat{p}(a_0)$ versus Jaccard distance. Cases with
    $\hat{p}(a_0) = 0$ are omitted.}
    \label{fig:e2scatter}
\end{figure}
 
Figure~\ref{fig:e2scatter} shows the collision rate of each sampler at several
thresholds. Within every threshold the relation is monotone---lower collision
rate, higher dispersion---with no exception among the seven samplers. The rank
correlation is around $-0.99$ ($n = 6$, $p < 0.001$) and stays there across the
whole range of thresholds at which all samplers have a positive collision rate.
Proposition~\ref{prop:carowei} therefore supplies not only a performance
guarantee but also a usable indicator, since lowering $p_\mu(\rho)$ does help
to achieve a higher Jaccard distance. 

\subsection{Comparison with existing methods}
\label{subsec:headtohead}

\begin{table}[ht]
\caption{Minimum pairwise Jaccard distance $D(\Pi)$ on 20 random 40-city instances; higher is better. Each column is a budget $c$ with the corresponding $k$. Entries are medians over the 20 instances, and for our sampler additionally over three seeds. The last row gives the number of instances on which our sampler strictly beats the Bi-Criteria Heuristic, the strongest baseline.}
\label{tab:rand40}
\centering
\begin{threeparttable}
\begin{tabular}{lrrrr}
\toprule
 & $c=2.0$ & $c=2.5$ & $c=3.0$ & $c=4.0$ \\
 & $k=80$ & $k=120$ & $k=160$ & $k=320$ \\
\midrule
Bi-Criteria & 0.8060 & 0.8410 & 0.8570 & 0.8240 \\
NMA & 0.7690 & 0.7500 & 0.7300 & 0.6210 \\
RF-MA3S & 0.0000 & 0.0000 & 0.0000 & 0.0000 \\
\midrule
GPN-hinge (c=2.0) & \textbf{0.8410} & \textbf{0.8570} & \textbf{0.8570} & \textbf{0.8410} \\
\midrule
Wins vs.\ Bi-Criteria & 20/20 & 20/20 & $19/20^{\dagger}$ & $20/20^{\ddagger}$ \\
\bottomrule
\end{tabular}
\begin{tablenotes}\footnotesize
\item[$\dagger$] 14 of them are ties.
\item[$\ddagger$] 2 of them are ties.
\end{tablenotes}
\end{threeparttable}
\end{table}

\paragraph{Random 40-city instances.} Table~\ref{tab:rand40} reports the comparison on the 20 random 40-city instances, the setting the sampler was trained for. At the training budget $c = 2.0$ the learned sampler attains a higher $D(\Pi)$ than all three baselines on every one of the 20 instances, and the same holds at $c = 2.5$. At $c = 4.0$ it wins on 18 of 20 against the Bi-Criteria Heuristic and on all 20 against the other two. The one setting where it does not separate is $c = 3.0$, where it ties with the Bi-Criteria Heuristic at the median and wins outright on 5 instances.

Two features of the baselines are worth recording. RF-MA3S returns
$D(\Pi) = 0$ at every budget, meaning the selected set contains a pair of identical tours; on the similarity scale this appears as a value of $1.0$, which understates how complete the collapse is. NMA degrades steadily as $k$ grows, from $0.769$ at $k = 80$ to $0.621$ at $k = 320$, consistent with a niching mechanism that cannot maintain separation once many solutions are requested.

\begin{table}[t]
\caption{Minimum pairwise Jaccard distance $D(\Pi)$ on TSPLIB instances at $c=4.0$ and $k=320$; higher is better. Our sampler is trained once on random 40-city instances at $c=2.0$ and applied zero-shot. Values are medians over three seeds; the seed range is $0.0000$ in every cell.}
\label{tab:tsplib-main}
\centering
\begin{tabular}{lrrrrr}
\toprule
Instance & $n$ & Bi-Criteria & NMA & RF-MA3S & GPN-hinge (c=2.0) \\
\midrule
\texttt{eil51} & 51 & 0.8539 & 0.6027 & 0.0000 & \textbf{0.8667} \\
\texttt{berlin52} & 52 & 0.8444 & 0.6667 & 0.0000 & \textbf{0.9053} \\
\texttt{st70} & 70 & 0.8333 & 0.6408 & 0.0000 & \textbf{0.8976} \\
\texttt{eil76} & 76 & 0.8571 & 0.6549 & 0.0000 & \textbf{0.9143} \\
\texttt{pr76} & 76 & 0.8308 & 0.6549 & 0.0000 & \textbf{0.9220} \\
\texttt{rat99} & 99 & 0.8488 & 0.6887 & 0.0000 & \textbf{0.9468} \\
\texttt{kroA100} & 100 & 0.8372 & 0.7013 & 0.0000 & \textbf{0.9362} \\
\texttt{kroB100} & 100 & 0.8439 & 0.7013 & 0.0000 & \textbf{0.9418} \\
\texttt{kroC100} & 100 & 0.8372 & 0.6842 & 0.0000 & \textbf{0.9362} \\
\texttt{kroD100} & 100 & 0.8372 & 0.6928 & 0.0000 & \textbf{0.9418} \\
\texttt{kroE100} & 100 & 0.8439 & 0.6842 & 0.0392 & \textbf{0.9418} \\
\texttt{rd100} & 100 & 0.8439 & 0.6928 & 0.0000 & \textbf{0.9247} \\
\texttt{eil101} & 101 & 0.8652 & 0.7051 & 0.0000 & \textbf{0.9312} \\
\midrule
Median & & 0.8439 & 0.6842 & 0.0000 & \textbf{0.9312} \\
\bottomrule
\end{tabular}
\end{table}

\begin{table}[t]
\caption{Summary across budgets. Medians are taken over the subset of instances on which our sampler returns $k$ feasible tours; the last column gives wins against Bi-Criteria, NMA and RF-MA3S on that subset.}
\label{tab:tsplib-budgets}
\centering
\begin{tabular}{rrrrrrrl}
\toprule
$c$ & $k$ & Achieved & GPN-hinge (c=2.0) & Bi-Criteria & NMA & RF-MA3S & Wins \\
\midrule
2.0 & 80 & 1/13 & \textbf{0.8667} & 0.8276 & 0.8000 & 0.0000 & 1/1, 1/1, 1/1 \\
2.5 & 120 & 5/13 & \textbf{0.9053} & 0.8571 & 0.7742 & 0.0000 & 5/5, 5/5, 5/5 \\
3.0 & 160 & 7/13 & \textbf{0.9220} & 0.8696 & 0.7541 & 0.0000 & 6/7, 7/7, 7/7 \\
4.0 & 320 & 13/13 & \textbf{0.9312} & 0.8439 & 0.6842 & 0.0000 & 13/13, 13/13, 13/13 \\
\bottomrule
\end{tabular}
\end{table}

\begin{table}[t]
\caption{Acceptance rate of the sampler, i.e.\ the fraction of the 1{,}000 sampled tours meeting the budget (median over three seeds). The sampler is trained at $n=40$ and applied without fine-tuning.}
\label{tab:tsplib-accept}
\centering
\begin{tabular}{lrrrrr}
\toprule
Instance & $n$ & $c=2.0$ & $c=2.5$ & $c=3.0$ & $c=4.0$ \\
\midrule
\texttt{eil51} & 51 & 0.352 & 0.997 & 1.000 & 1.000 \\
\texttt{berlin52} & 52 & 0.002 & 0.423 & 0.990 & 1.000 \\
\texttt{st70} & 70 & 0.000 & 0.489 & 0.995 & 1.000 \\
\texttt{eil76} & 76 & 0.001 & 0.687 & 1.000 & 1.000 \\
\texttt{pr76} & 76 & 0.000 & 0.001 & 0.425 & 1.000 \\
\texttt{rat99} & 99 & 0.000 & 0.000 & 0.004 & 0.986 \\
\texttt{kroA100} & 100 & 0.000 & 0.000 & 0.000 & 0.706 \\
\texttt{kroB100} & 100 & 0.000 & 0.000 & 0.000 & 0.909 \\
\texttt{kroC100} & 100 & 0.000 & 0.000 & 0.000 & 0.732 \\
\texttt{kroD100} & 100 & 0.000 & 0.000 & 0.000 & 0.842 \\
\texttt{kroE100} & 100 & 0.000 & 0.000 & 0.000 & 0.877 \\
\texttt{rd100} & 100 & 0.000 & 0.004 & 0.632 & 1.000 \\
\texttt{eil101} & 101 & 0.000 & 0.232 & 0.994 & 1.000 \\
\bottomrule
\end{tabular}
\end{table}

\paragraph{TSPLIB instances.} Table~\ref{tab:tsplib-main} reports the same comparison on the 13 TSPLIB instances with 51 to 101 cities, using the identical checkpoint with no fine-tuning. At $c = 4.0$ the learned sampler attains a higher $D(\Pi)$ than all three baselines on all 13 instances, raising the median from $0.844$ for the Bi-Criteria Heuristic to $0.931$. On the similarity scale this is a reduction of the maximum pairwise Jaccard index from $0.156$ to $0.069$, a factor of $2.17$. RF-MA3S again returns $D(\Pi) = 0$ on 12 of the 13 instances.

Table~\ref{tab:tsplib-budgets} summarizes the four budgets. Of the 26 comparisons in which the sampler returns $k$ feasible tours, it wins 25 and ties once, against the Bi-Criteria Heuristic on \texttt{eil51} at $c = 3.0$. At the tighter budgets, however, it does not always return $k$ feasible tours: $1/13$ instances at $c = 2.0$, rising to $13/13$ at $c = 4.0$.

Table~\ref{tab:tsplib-accept} shows why. The acceptance rate falls
monotonically with instance size, from $0.35$ at $n = 51$ to zero for $n \ge 70$ at the tightest budget. The sampler was trained at $n = 40$ and applied without adaptation, and the cost of the tours it produces degrades with scale, so larger instances need a wider budget before $k$ feasible tours survive the filter or retrain at the corresponding scale. We leave the extension to large instances to future work.

\section{Conclusion}
\label{sec:conclusion}

We have studied D-TSP through a sample-then-select framework, and located its difficulty in the sampling stage. We show that the optimal sampling distribution is learnable via a closed form function and provide an analysis for end-to-end guarantee. To decrease the difficulty of learning the diverse tours directly, we provide a pipeline to learn the diverse spanning trees and convert them to tours, which costs nothing according to our analysis. Lastly, we provide another end-to-end guarantee to ensure the sampling pool merely hold some $k$ tours lying far apart. Turns an exponential sample size into a linear one, at the cost of a weaker benchmark. Experiments show that the resulting sampler, trained once and applied without fine-tuning, beats existing methods both in distribution and on TSPLIB, provided the evaluation budget matches the one it was trained for.

%\THEEndNotes
% \begingroup \parindent 0pt \parskip 0.0ex \def\enotesize{\normalsize} \theendnotes \endgroup

% Appendix here

%% ------------------------------------------------------------------
%% DOUBLE-ANONYMOUS REVIEW: TELO requires that all acknowledgments be
%% removed from the submitted version. The funding statement below names
%% the grants and therefore identifies the institution -- keep it
%% commented out until the camera-ready.
%%
%% \begin{acks}
%% This research was supported by the National Science and Technology
%% Council, Taiwan, under grant numbers NSTC 114-2221-E-008-049 and
%% NSTC 111-2222-E-008-008-MY2.
%% \end{acks}
%% ------------------------------------------------------------------

\bibliographystyle{ACM-Reference-Format}
\bibliography{reference}

% =====================================================================
% Diverse TSP -- Appendix
%
%   \input{appendix}
% to main.tex, AFTER the bibliography (ACM convention).
%
% Notation note: the ring radius of the lower-bound instance is rho,
% because R denotes the reward. eta is the per-cell cost gap and is used
% nowhere else in the paper.
% =====================================================================

\appendix

\section{Uniform sampling of tours}
\label{app:rejection}

We prove the claim made in Section~\ref{subsec:hit}, that rejection sampling
cannot produce the uniform sampler on $\Omega_c$. Let $V$ consist of $n$ points
drawn i.i.d.\ uniformly from $[0,1]^2$ and let $\pi$ be drawn uniformly from
$\mathcal{T}_n$, independently of $V$.

Conditioned on $V$, the tour $\pi$ visits the points in a uniformly random
order, so each of its $n$ edges is a uniformly random pair and
$\mathbb{E}[C(\pi) \mid V] = n\,\mathbb{E}\|U - U'\| \approx 0.5214\,n$, where
$U$ and $U'$ are independent uniform points. The Beardwood--Halton--Hammersley
theorem gives $C^{*} \sim \beta \sqrt{n}$ almost surely with
$\beta \approx 0.7124$ \citep{applegate2006traveling}. On that almost sure event
the ratio $\mathbb{E}[C(\pi) \mid V] / C^{*}$ grows like $0.73 \sqrt{n}$, so for
any fixed $c$ there is an $n_0(c)$ beyond which
$\mathbb{E}[C(\pi) \mid V] \ge 2c\,C^{*}$.

It remains to show that $C(\pi)$ concentrates. Its randomness is that of a
uniform permutation, whose coordinates are not independent, so the usual
bounded-difference inequality does not apply directly; we use its version for
uniform permutations. Transposing two positions of the permutation changes at
most four edges of the tour, and every edge has length at most $\sqrt 2$, so the
transposition changes $C(\pi)$ by at most $4\sqrt 2$. The permutation version of
the bounded-difference inequality therefore gives, for a universal constant
$\kappa$,
\[
  \Pr\bigl[\, C(\pi) \le \mathbb{E}[C(\pi) \mid V] - t \,\bigm|\, V \,\bigr]
  \;\le\; \exp\bigl(-\kappa t^{2}/n\bigr) .
\]
Taking $t = \mathbb{E}[C(\pi) \mid V] - c\,C^{*} = \Theta(n)$ yields
$\Pr[\pi \in \Omega_c \mid V] \le e^{-\Theta(n)}$ for all $n \ge n_0(c)$, on an
event of $V$ of probability tending to one.

\section{The lower-bound instance}
\label{app:gadget}

We give the coordinates behind Figure~\ref{fig:gadget} and prove
Theorem~\ref{thm:lower}.

\paragraph{Construction.}
Fix $g \ge 3$. For $j = 1,\dots,g$ let $\phi_j = 2\pi j / g$, let $\hat u_j$ and
$\hat v_j$ be the outward radial and the tangential unit vectors at angle
$\phi_j$, and let $P_j = \rho\,\hat u_j$. Given $0 < H < L \ll \rho/g$ and an
offset $\vartheta \in (0,H)$, the instance consists of the $n = 4g$ points
\[
  p^1_j = P_j - L\hat v_j, \quad
  p^2_j = P_j, \quad
  p^3_j = P_j + L\hat v_j, \quad
  z_j = P_j - \vartheta\,\hat v_j + H\hat u_j ,
\]
with consecutive cells joined by the edge $\{p^3_j, p^1_{j+1}\}$. A tour of the
shape of Figure~\ref{fig:gadget} therefore uses three edges inside each cell and
$g$ connecting edges, which is $4g = n$ in total, as it must be by
Equation~\ref{eq:equal-cardinality}. The two traversals share only
$\{z_j, p^2_j\}$, and subtracting their costs gives Equation~\ref{eq:eta}.

\paragraph{Structure.}
For $\vartheta$ small enough, every tour of cost at most $C^{*} + g\eta$ visits
the cells in ring order and traverses each cell by one of the two ways of
Figure~\ref{fig:gadget}. Indeed, tour cost is $2$-Lipschitz in $\vartheta$,
while at $\vartheta = 0$ all such structured tours have equal cost and every
unstructured tour is more expensive by at least some $\varsigma > 0$ depending
only on $(\rho, L, H, g)$; taking $\vartheta$ with $2\vartheta < \varsigma$
gives the claim. Writing $\mathbf{b} \in \{0,1\}^g$ for the vector of per-cell
choices, with $b_j = 1$ meaning the expensive traversal, $|\mathbf b|$ for its
number of ones, and $h$ for the number of coordinates in which two such vectors
differ, we obtain
\begin{equation}
  C(\pi_{\mathbf b}) = C^{*} + |\mathbf b|\,\eta,
  \qquad
  a(\pi_{\mathbf b}, \pi_{\mathbf b'}) = n - 2h ,
  \label{eq:gadget-laws}
\end{equation}
since the two traversals of a cell share one of their three internal edges and
the connecting edges are common to all such tours. We verified both laws
exhaustively for $(\rho, L, H) = (20, 1, 0.5)$ and $g = 3$, enumerating all
tours on the $n = 12$ points for
$\vartheta \in \{2\times10^{-2}, 5\times10^{-3}, 10^{-3}\}$; the margin
$\varsigma/\eta$ enforcing the structure grew from $32$ to $632$ across that
range.

\paragraph{Proof of Theorem~\ref{thm:lower}.}
Choose $c$ with $(c-1)C^{*} = \eta$. By Equation~\ref{eq:gadget-laws} the
feasible region is $\Omega_c = \{\pi_{\mathbf 0}\} \cup
\{\pi_{\mathbf e_j}\}_{j \le g}$, so $N_c = g+1$, the tour $\pi_{\mathbf 0}$
costs $C^{*}$ and the remaining $g$ cost exactly $c\,C^{*}$.

Equations~\ref{eq:gadget-laws} and~\ref{eq:J-in-a} then give
$d(\pi_{\mathbf e_i}, \pi_{\mathbf e_j}) = 2/(g+1)$ for $i \ne j$ and
$d(\pi_{\mathbf 0}, \pi_{\mathbf e_j}) = 2/(2g+1)$, using $n = 4g$. With
$k = g$ the only set avoiding $\pi_{\mathbf 0}$ is the set of all $g$ expensive
tours, which attains $D^{*}_g(\Omega_c) = 2/(g+1)$; any other choice of $g$
feasible tours contains $\pi_{\mathbf 0}$ and has $D(\Pi) = 2/(2g+1)$, that is,
$\tfrac{g+1}{2g+1} D^{*}_g(\Omega_c)$.

The first clause of the theorem is now Lemma~\ref{lem:sample} with
$q = 1/(g+1)$ and $k = g$. For the second, let $\mu$ be any sampler on
$\Omega_c$ whose smallest mass falls on an expensive tour, say
$\pi_{\mathbf e_1}$. By Equation~\ref{eq:delta} that mass equals
$1/\bigl((g+1)\Delta(\mu)\bigr)$, so the probability that $M$ draws all miss
$\pi_{\mathbf e_1}$ is at least $1 - M/\bigl((g+1)\Delta(\mu)\bigr)$, which is at
least $1/2$ whenever $M \le (g+1)\Delta(\mu)/2$. On that event the pool is
contained in $\{\pi_{\mathbf 0}, \pi_{\mathbf e_2}, \dots, \pi_{\mathbf e_g}\}$,
which has exactly $g$ elements; any rule returning $g$ distinct tours therefore
returns that set, whose value was bounded above, and if the pool holds fewer
than $g$ distinct tours no rule returns $g$ of them at all.

\paragraph{The Gibbs sampler on this instance.}
Let $R$ be any reward with $R(\pi_{\mathbf 0}) - R(\pi_{\mathbf e_j}) = \Gamma >
0$ for every $j$, as holds for instance when $R = -C$, with $\Gamma = \eta$. The
unnormalized weights of Proposition~\ref{prop:gibbs} are then $e^{\Gamma/\alpha}$
for $\pi_{\mathbf 0}$ and $1$ for each expensive tour, so the smallest mass does
fall on an expensive tour and
\[
  \Delta(\mu_\alpha) \;=\; \frac{e^{\Gamma/\alpha} + g}{g+1} ,
\]
which takes every value above $1$ as $\alpha$ ranges over $(0,\infty)$. The
hypothesis of the second clause of Theorem~\ref{thm:lower} is thus met by every
reward that strictly prefers the cheap tour, at every temperature.

\section{Network architecture}
\label{app:architecture}

The network computing the conditionals of
Equation~\ref{eq:autoregressive-policy} follows \citet{ma2019combinatorial} and
is not a contribution of this work. We record it here for completeness, and
Figure~\ref{fig:model} shows its layout.

\begin{figure}[htbp]
    \centering
    \includegraphics[width=\linewidth]{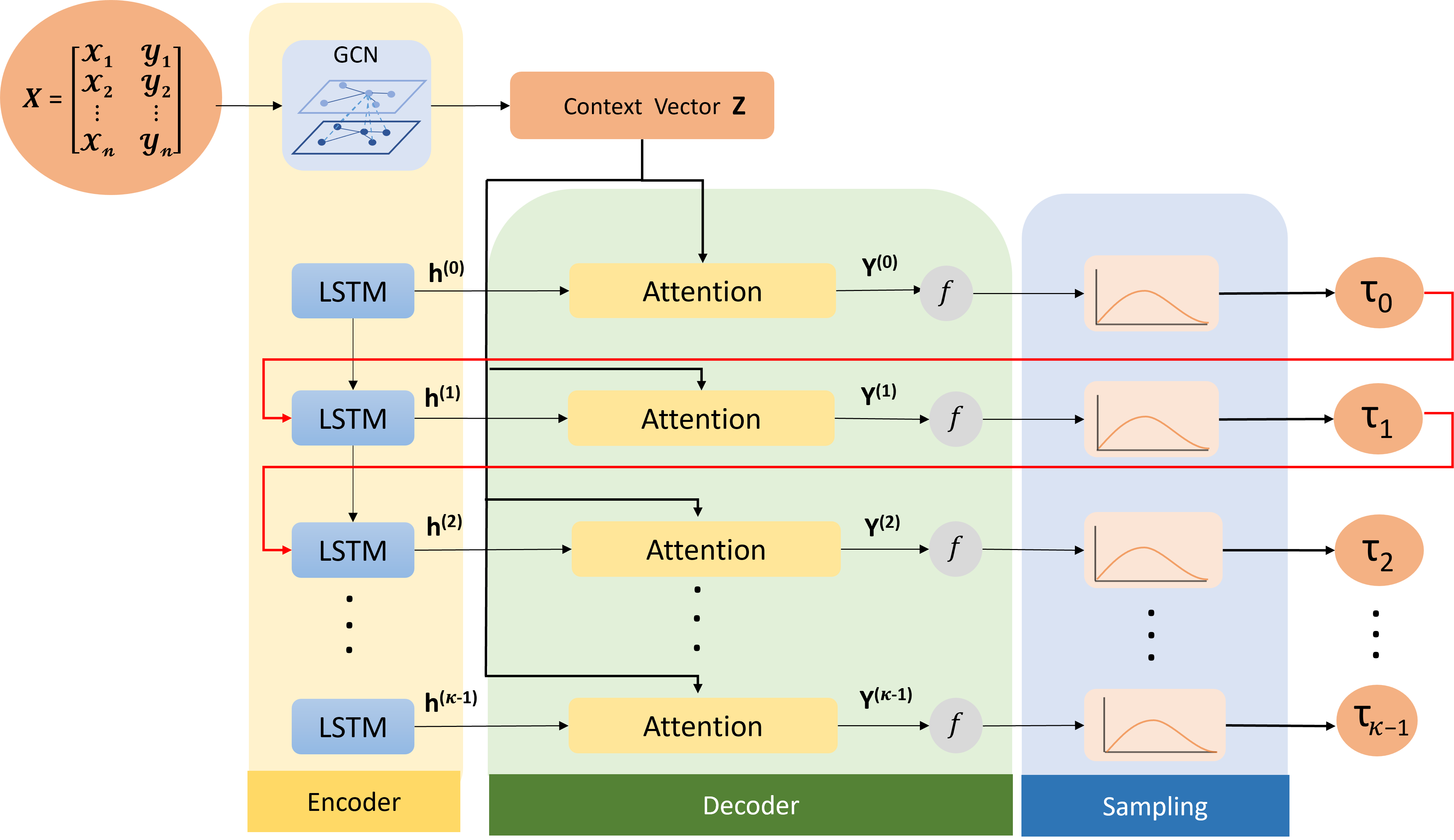}
    \caption{Overview of the network architecture.}
    \label{fig:model}
\end{figure}

Let $\mathbf{X} \in \mathbb{R}^{n \times 2}$ be the matrix of vertex
coordinates. An encoder produces contextual embeddings: the coordinates are
first embedded, then updated over $\ell$ layers of residual message passing,
\begin{equation}
    x_u^{(i)} = \gamma\, x_u^{(i-1)}\Theta
      + (1-\gamma)\,\phi_\theta\!\left(
        \frac{1}{|\mathcal{N}(u)|}
        \bigl\{ x_v^{(i-1)} \bigr\}_{v \in \mathcal{N}(u) \cup \{u\}}
      \right),
    \label{eq:gcn}
\end{equation}
where $\phi_\theta$ is the aggregation function, $\Theta$ a trainable weight
matrix, and $\gamma$ a trainable coefficient balancing a vertex's own previous
representation against its aggregated neighbourhood. Edge embeddings are formed
as $\mathbf{z}_{uv} = \varepsilon_{uv} + x_u^{(\ell)} + x_v^{(\ell)}$, with
$\varepsilon_{uv}$ the initial edge embedding. An LSTM tracks the sequential
state: it reads $\mathbf{z}_{uv}$ for the previously selected edge
$(u,v) = \tau_{t-1}$ and emits a query vector $\mathbf{q}^{(t)}$.

The decoder scores each edge by
\begin{equation}
Y_{uv}^{(t)} = \mathbf{v}^{\!\top} \tanh\bigl(
  \mathbf{W}_{\mathrm{key}}\mathbf{z}_{uv}
  + \mathbf{W}_{\mathrm{query}}\mathbf{q}^{(t)} \bigr),
\label{eq:attention}
\end{equation}
with $\mathbf{W}_{\mathrm{key}}$ and $\mathbf{W}_{\mathrm{query}}$ trainable
matrices and $\mathbf{v}$ a learnable vector. The scores are scaled by a
constant $s$, masked to exclude inadmissible edges, and normalized:
$P^{(t)} = \mathrm{softmax}(\mathrm{mask}(s \cdot Y^{(t)}))$. The embedding of
the selected edge is fed back to the LSTM, and the two modules alternate until
the sequence is complete. The matching policy of Method~2 uses the same network
with $\kappa = |O|/2$ over the odd-degree vertex set $O$ and a mask admitting an
edge whenever both endpoints are still unmatched.

\end{document}